%% file: paper.tex
\documentclass[sigconf, nonacm]{acmart}
\newcommand\vldbdoi{XX.XX/XXX.XX}

\newcommand\vldbavailabilityurl{}

\usepackage{graphicx}
\usepackage{tabularx}
\usepackage{algorithm}
\usepackage{algorithmicx}
\usepackage[noend]{algpseudocode}

\usepackage{amstext}
\usepackage{amsthm}

\usepackage{blindtext}
\usepackage{color}
\usepackage{siunitx}
\DeclareSIUnit[number-unit-product = ]\percent{\char`\%}

\usepackage{listings}
\usepackage{multirow}

\newtheorem{thm}{Theorem}

\newcommand{\textred}[1]{\textcolor{red}{#1}}

\ifx\noeditingmarks\undefined
   \newcommand{\pgwrapper}[2]{\textbf{#1: }\textred{\textit{#2}}}

\else
   \newcommand{\pgwrapper}[2]{}

\fi

\newcommand{\cuttext}[1]{}

\definecolor{truegreen}{RGB}{9, 155, 9}

\definecolor{fbblue}{RGB}{67, 96, 156}
\newcommand{\fbok}[1]{#1} 

\newcommand{\fb}{Meta}

\newcommand{\fbprom}{Meta}

\newcommand{\thefb}{Meta}

\newcommand{\CCC}{Cloned Concurrency Control}
\newcommand{\ccc}{cloned concurrency control}

\newcommand{\sysName}{C5\xspace}
\newcommand{\sys}{C5}
\newcommand{\cicada}{Cicada}
\newcommand{\myrocks}{MyRocks}
\newcommand{\syscicada}{\sys{}-\cicada{}}
\newcommand{\sysmyrocks}{\sys{}-\myrocks{}}

\newcommand{\kuafu}{KuaFu\xspace}

\newcommand{\scheduler}{scheduler}

\newcommand{\worker}{worker}
\newcommand{\Worker}{Worker}
\newcommand{\snapshotter}{snapshotter}
\newcommand{\Snapshotter}{Snapshotter}

\usepackage{titlesec}

\titlespacing*{\section}{0pt}{1ex plus 1ex minus .2ex}{1ex plus .2ex}

\titlespacing*{\subsection}{0pt}{1ex plus 1ex minus .2ex}{1ex plus .2ex}

\renewenvironment{enumerate}{
  \begin{list}{\labelenumi}{
     \usecounter{enumi}
     \setlength{\topsep}{0.5ex}
     \setlength{\itemsep}{0pt}
     \setlength{\itemindent}{0pt}
     \setlength{\leftmargin}{\labelwidth}
     \addtolength{\leftmargin}{-2pt}}
}{\end{list}}

\begin{document}

\title{\sys{}: Cloned Concurrency Control That Always Keeps Up}

\author{Jeffrey Helt}
\affiliation{%
  \institution{Princeton University}
}
\email{jhelt@cs.princeton.edu}

\author{Abhinav Sharma}
\affiliation{%
  \institution{Meta Platforms}
}
\email{abhinavsharma@fb.com}

\author{Daniel J. Abadi}
\affiliation{%
  \institution{University of Maryland, College Park}
}
\email{abadi@cs.umd.edu}

\author{Wyatt Lloyd}
\affiliation{%
  \institution{Princeton University}
}
\email{wlloyd@princeton.edu}

\author{Jose M. Faleiro}
\affiliation{%
  \institution{Microsoft Research}
}
\email{jmf@microsoft.com}

\begin{abstract}
  Asynchronously replicated primary-backup databases are commonly deployed to
  improve availability and offload read-only transactions. To both apply replicated
  writes from the primary and serve read-only transactions, the backups
  implement a \ccc{} protocol. The protocol ensures read-only transactions always
  return a snapshot of state that previously existed on the primary. This
  compels the backup to exactly copy the commit order resulting from the
  primary's concurrency control. Existing \ccc{} protocols guarantee this by
  limiting the backup's parallelism. As a result, the primary's concurrency
  control executes some workloads with more parallelism than these protocols. In
  this paper, we prove that this parallelism gap leads to unbounded replication
  lag, where writes can take arbitrarily long to replicate to the backup and
  which has led to catastrophic failures in production systems. We then design
  \sys{}, the first cloned concurrency protocol to provide bounded
  replication lag. We implement two versions of \sys{}: Our evaluation in \myrocks{}, a widely deployed database, demonstrates \sys{} provides bounded replication lag. Our evaluation in \cicada{}, a recent in-memory database, demonstrates \sys{} keeps up with even the fastest of primaries.
\end{abstract}




\maketitle



\input{intro}
\input{bg}
\input{scheduling}

\input{design}
\input{myrocks_impl}
\input{myrocks_evaluation}
\input{cicada_impl}
\input{deployment}
\input{related}

\input{conclusion}



\clearpage
\bibliographystyle{ACM-Reference-Format}
\bibliography{ref}

\clearpage




\end{document}

%% file: intro.tex
\section{Introduction}
\label{sec:intro}

Asynchronously replicated primary-backup databases are the cornerstones of many
applications~\cite{existential2015sosp,instagram2015scaling,espresso2013sigmod,verbitski2018aurora,antonopoulos2019socrates}.
In these systems, after the primary executes a transaction, it sends the
resultant writes to a set of backups. The backups apply the writes to
reconstruct the primary's state and execute read-only transactions against their
local state. To simultaneously execute writes and read-only transactions,
a backup implements a \textit{\ccc{}} protocol. In addition to providing
availability if the primary fails, these protocols improve the database's
performance: throughput is increased by serving reads from many backups and
latency is reduced by serving reads from a nearby backup.




To reap these benefits without breaking overlying applications, a \ccc{}
protocol must guarantee \textit{monotonic prefix consistency}, where it exposes a
progressing sequence of the primary's recent states to read-only transactions.
This ensures backups never return values from states that did not exist on
the primary, thereby helping maintain application invariants.

But monotonic prefix consistency makes no guarantees about how quickly writes
replicate to the backup. In theory, they could be delayed indefinitely. To be
reliable, a \ccc{} protocol must also guarantee \textit{bounded replication lag}.
Intuitively, a transaction's replication lag is the time between when its
writes are first observable by reads on the primary and backup. By guaranteeing
bounded replication lag, a \ccc{} protocol ensures transactions always appear
shortly.

Guaranteeing bounded replication lag is important; significant lag has led to
catastrophic failures. For instance, GitLab was unavailable for eighteen hours after a
workload change caused such significant lag that replication stopped entirely.
In the process of fixing the issue, user data was lost~\cite{gitlab2017incident,
  gitlab2017postmortem}. \fbok{Similarly, several times in past years, \fbprom{}
  routed all user requests away from a data center because too many of that location's
  backups had excessive lag.}

To guarantee bounded replication lag, a \ccc{} protocol must apply
writes with as much parallelism as was used by the primary's concurrency control
protocol. But guaranteeing monotonic prefix consistency severely constrains the
\ccc{} protocol, making it difficult to execute with sufficient parallelism. For
example, consider two concurrent transactions with both conflicting and
non-conflicting writes. If the primary employs two-phase locking~\cite{bernstein-book}, the non-conflicting writes can execute
in parallel, and the commit order is determined by the lock
acquisition order on the first conflicting write. Once their commit order is
chosen, however, monotonic prefix consistency mandates that a backup's state
reflects it. Thus, the \ccc{} protocol must ensure the transactions are
serialized correctly, potentially constraining its parallelism. 

In the past, slow I/O devices bottlenecked the primary and backup, dominating differences in parallelism. But low-latency persistent storage and large main memories
removed this bottleneck, so the primary's concurrency control and the backup's
\ccc{} protocols are now directly competing.



Existing protocols differ in how much parallelism they leverage while
executing writes. On one end of the spectrum are single-threaded
protocols~\cite{postgres,mysql}. On the other are
transaction-~\cite{hong2013kuafu,king1991remote,oracle2001txnscheduler,mysql8writeset}
and
page-granularity~\cite{verbitski2018aurora,oracle2019adg,antonopoulos2019socrates}
protocols. The former execute non-conflicting transactions in parallel; the
latter execute writes to different pages in parallel. For both, however, there
are workloads where the primary's concurrency control always executes
with more parallelism. Using these workloads, we prove neither class of
protocols guarantees bounded replication lag.
In turn, implementations of these classes of protocols
are not reliable because changes to the workload or the primary's
concurrency control can suddenly lead to unbounded replication lag.

\input{contribution_table}

In this paper, we present \sys{}, the first \ccc{} protocol to provide
bounded replication lag. To always keep up, \sys{}'s
insight is that the backup's protocol must execute writes 
at the same granularity as the primary's concurrency control
protocol.
Thus, its \underline{c}loned \underline{c}oncurrency \underline{c}ontrol has \underline{c}ommensurate \underline{c}onstraints (\sys{}) with the primary.
Because the primary executes writes to non-conflicting rows in parallel, \sys{} uses a row-granularity protocol.

But row-granularity execution introduces several challenges. First, applying
individual row writes to the backup's state can lead to permanent violations of
monotonic prefix consistency, where the backup's state ceases to match the
primary's. To avoid such violations, \sys{}'s \scheduler{} calculates the necessary metadata for
its \worker{}s to correctly order writes to each row. Second, row-granularity execution does not
guarantee monotonic prefix consistency for read-only transactions because
transactional atomicity and commit order are not necessarily respected. Imposing
additional constraints on \worker{}s, however, could reintroduce
replication lag. Instead, \sys{}'s \snapshotter{} uses three progressing
snapshots, ensuring reads observe a consistent state without constraining
execution.

We show formally that a
row-granularity protocol never imposes more constraints on the backup's
execution than a valid concurrency control protocol imposes on the primary's.
Thus, \sys{} can, in theory, always match the primary's parallelism.

In practice, however, row-granularity execution is necessary but not sufficient to provide bounded replication lag. Other bottlenecks, such as a slow \scheduler{}, may get in the way. We thus
implement two versions of \sys{}, \sysmyrocks{} and \syscicada{}, to confirm it always keeps up.
\sysmyrocks{} is backward-compatible and deployed in production at \fb{}.
Making it backward-compatible, however, required some additional constraints to the parallelism in our design.
We thus also implemented \syscicada, which faithfully implements our design (without additional constraints) and demonstrates \sys{} can keep up with a cutting-edge concurrency control protocol.

We compare each of our \sys{} implementations to a state-of-the-art, transaction-granularity protocol~\cite{hong2013kuafu}.
We find the transaction-granularity protocol can keep up on some workloads.
But unbounded replication lag is lurking nearby: Simple optimizations that improve primary throughput cause the backup to lag.
In contrast, our \sys{} implementations always keep up.

Figure~\ref{tbl:contributions} summarizes the primary contributions of
this paper. In addition, Section~\ref{sec:deployment} describes
experience from deploying \sys{} at Meta. Our experience echoes our
evaluation. The simple single-threaded cloned currency control that
was previously deployed could often keep up with the primary. But
large replication lag would be exposed by workload
changes. The deployment of \sys{} eradicated these issues and led to
noticeably better reliability.




%% file: contribution_table.tex
\begin{table}[t]
  \small
  \centering
\begin{tabular}{@{}l l l@{}}
  \midrule
  \textbf{Execution}    & \textbf{Parallelism}     & \textbf{Implementation That}\\
  \textbf{Constraints}  & \textbf{$\ge$ Primary?}  & \textbf{Always Keeps Up}\\
  \hline\\[-1.8ex]
  Existing Protocols &&\\
  \hspace{1ex} Transaction-Granularity        & $\times$ (\S\ref{sec:scheduling:proof})            & $\implies$Impossible\\
  \hspace{1ex} Page-Granularity       & $\times$ (\S\ref{sec:scheduling:page-granularity}) & $\implies$Impossible\\
  \sysName{} Protocols &&
  \multirow{3}{*}{\shortstack[l]{\sysmyrocks{} (\S\ref{sec:myrocks-impl})\\ \syscicada{} (\S\ref{sec:cicada})}}\\
    \hspace{1ex} Primary Granularity & $\checkmark$ (\S\ref{sec:design:scheduler-proof}) & \\[1.5ex]
  \midrule \vspace{0.5mm}
\end{tabular}
  \caption{Contribution summary. We prove existing \ccc{} protocols cannot keep up because their execution granularity exposes less parallelism than their primary. We introduce primary-granularity execution with \sys{}, show it exposes as much parallelism as the primary, and evaluate two implementations to show they always keep up.}
  \label{tbl:contributions}
\end{table}

%% file: bg.tex
\section{Background}
\label{sec:bg}

This section gives background on primary-backup databases, \ccc{} protocols,
and their guarantees.

\subsection{Motivating Example}
\label{sec:bg:example}

Throughout the paper, we use the following motivating example.
Alice, Bob, and Charlie use a social media platform to
share and comment on videos. The platform stores its videos and comments in a
database. One table stores each video's name and metadata, including a per-video
comment counter; a second table stores each comment's text and metadata. When a
user comments on a video, an application server executes a transaction of two
operations: it first inserts a new row in the comment table and then
increments the video's counter.

The platform replicates the primary's database at a set of
backups. The primary implements a concurrency control protocol, and each
backup implements the \ccc{} protocol.

\subsection{Primary-Backup Replication}
\label{sec:bg:replication}

\begin{figure}[tb]
  \centering
  \includegraphics[page=5,width=0.9\linewidth]{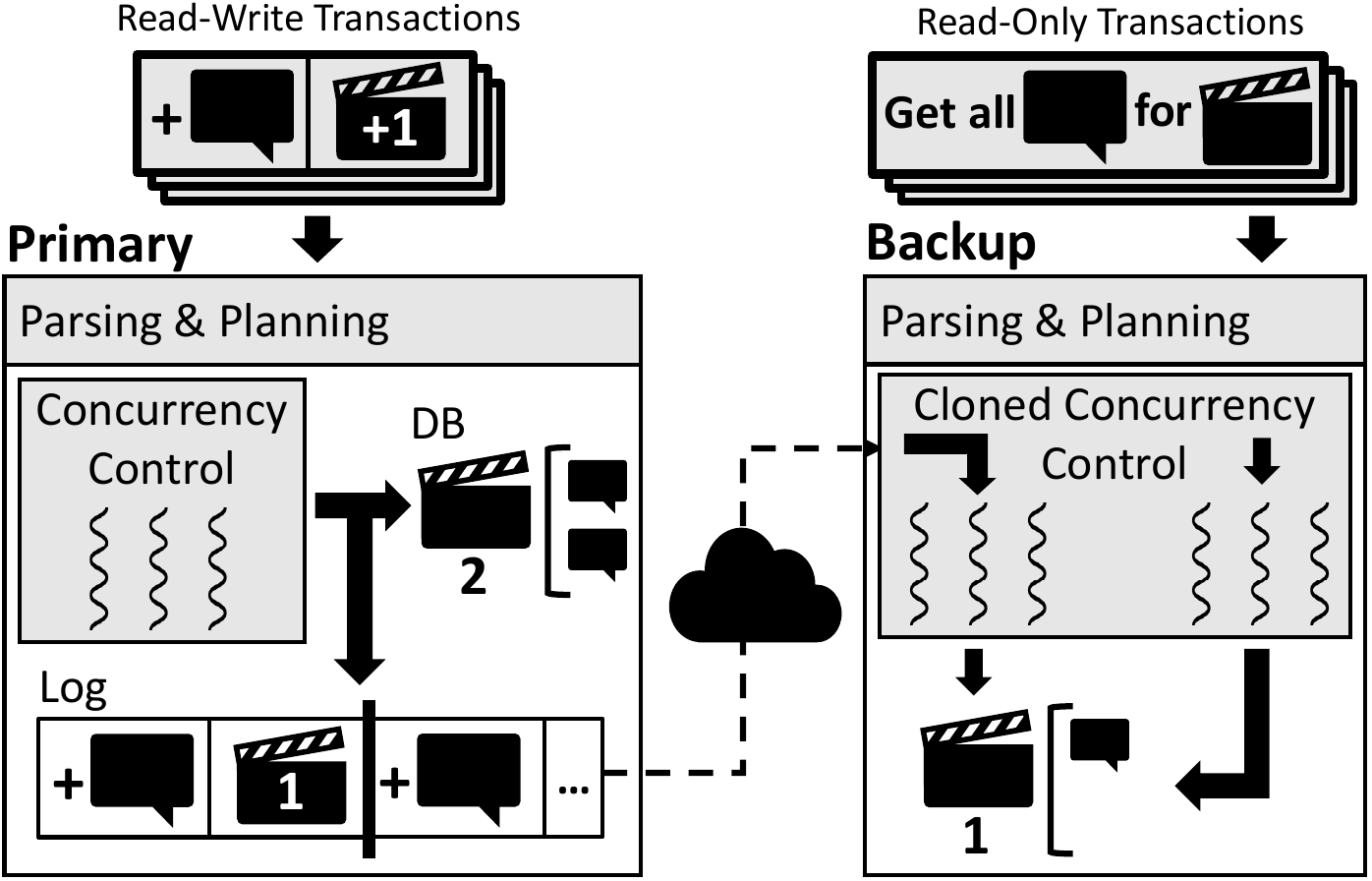}
  \caption{Transaction processing in primary-backup. The primary
    executes read-write transactions. The backup replicates the changes and
    executes read-only transactions.}
  \label{fig:primary-backup}
\end{figure}

Figure~\ref{fig:primary-backup} shows an overview of the primary and backup's
processing as they execute transactions. For each operation in a read-write
transaction, the primary parses it and plans its execution. Each plan, which may
include row queries, local computation, and row writes (i.e., inserts, updates,
and deletes), is then executed. For instance, to increment a video's comment
counter, the primary reads the counter's current value from the video's row in
the video table, increments it, and writes the result back to the row. After all
operations execute, the transaction commits by writing to the primary's database
and flushing a log of its changes to stable storage.

The primary then sends a copy of its log to the backup. The log reflects a total
order of the writes applied by the primary, determined by the primary's
transaction commit order and the order of each transaction's operations. The log
includes, for each transaction, the written rows and metadata to demarcate its
writes from those of
others~\cite{king1991remote,mysql,mysqlgroupcommit,mysql8writeset,myRocks,oracle2019ggate,mariadb10,hong2013kuafu,oracle2001txnscheduler}.

The backup's \ccc{} protocol reads the operations in the log and schedules them
for execution by \worker{} threads, bypassing parsing and planning.
The \worker{}s apply operations to the backup's copy of the database.
The protocol also executes read-only transactions using separate threads.

\subsection{Monotonic Prefix Consistency}
\label{sec:bg:consistency}

While an asynchronously replicated backup's state inevitably lags, it is ideally otherwise indistinguishable from the primary.
Intuitively, the backup should expose a progressing sequence of the primary's
recent states. This intuitive behavior is provided by many existing
systems~\cite{myRocks,mysql,postgres,verbitski2018aurora,oracle2019adg,antonopoulos2019socrates,oracle2001txnscheduler,mysql8writeset,mysqlgroupcommit,mariadb10,oracle2019ggate,hong2013kuafu,king1991remote,wang2017query}.
We refer to this guarantee here as \textit{monotonic prefix
  consistency} (MPC).

We define monotonic prefix consistency relative to the primary's log of transactions (and thus writes). MPC comprises two
guarantees: First, the backup's state must reflect the changes of a contiguous
prefix of transactions. Second, the sequence of states exposed by
the backup to read-only transactions must reflect prefixes of monotonically
increasing length.

In our example application, MPC ensures read-only transactions never see a
mismatch between the number of comments on a video and the video's comment
counter; each transaction's changes appear atomically. Further, MPC ensures
comments never seem to disappear. Once a comment becomes visible to a user, all future states exposed by the same backup will include it. Although beyond our scope here, MPC can be guaranteed across multiple backups using sticky sessions~\cite{terry1994sessions} or with client-tracked metadata.

MPC also maintains implicit application invariants. For instance, suppose Alice
first updates her default video permissions to share her future videos only with
Bob and then uploads a new video. To make these changes, a transaction first
updates her default access control list to only include Bob, and a
subsequent transaction adds the new video. An implicit invariant, that Charlie
should not see the new video, is expressed by the order of the two
transactions. MPC does not break such invariants because states always reflect
contiguous prefixes of the log.




\subsection{Bounded Replication Lag}
\label{sec:bg:replication-lag}

Monotonic prefix consistency specifies a \ccc{} protocol's correctness but
does not clarify its performance requirements. For instance, if Alice calls Bob
after commenting on his video, Bob should ideally see her new comment by the time he
receives her call. Given only MPC, the comment may be delayed for an arbitrarily
long time.

We define \textit{replication lag} as the time between when a transaction's
changes are included in the state returned by the primary and backup. (For the purposes of this paper, we assume the log is always delivered promptly to the backup.) More
precisely, we say a transaction $T$ is included in the state returned by the
primary or backup once either its writes or later writes are returned to reads.
To include a transaction in the returned state
requires the backup's protocol to do one of the following: (1) it can eagerly apply the transaction's changes
to its copy of the database, making them visible to future reads without
additional processing beyond that required to execute the read at the
primary~\cite{postgres,mysql,mysql8writeset,hong2013kuafu,king1991remote,oracle2001txnscheduler,verbitski2018aurora,oracle2019adg,antonopoulos2019socrates,mariadb10,oracle2019ggate};
or (2) it can defer part of the execution of the transaction's changes until a corresponding read arrives~\cite{wang2017query}.
For each $T$, we then define $f_p(T)$ and $f_b(T)$ as the real time when the
primary and backup respectively include $T$ in their state. For eager
protocols, $f_b(T)$ is the first time at which an arriving read would see $T$.
For lazy protocols, $f_b(T)$ is the first time at which an arriving read would
see $T$, plus the additional time required to finish any deferred execution.

A \ccc{} protocol guarantees \textit{bounded replication lag} if there exists
some finite time $L$ such that for all workloads $W$ and for all transactions
$T$ in $W$, $f_b(T) - f_p(T) \leq L$. (Transactions and workloads are defined
more precisely in Section~\ref{sec:scheduling:proof}.) In practice, guaranteeing
bounded lag ensures Bob never waits very long to see Alice's comment.


%% file: scheduling.tex
\section{Unbounded Lag In Existing Protocols}
\label{sec:scheduling}

Guaranteeing bounded replication lag is challenging. To satisfy MPC, the
backup's \ccc{} protocol must ensure the backup's state converges to the
primary's. To accomplish this, existing protocols serialize conflicting
writes~\cite{postgres,mysql,hong2013kuafu,king1991remote,oracle2001txnscheduler,mysql8writeset,verbitski2018aurora,oracle2019adg,antonopoulos2019socrates,mariadb10,oracle2019ggate,wang2017query}.

Serialization limits the backup's parallelism. But to guarantee
bounded replication lag, the backup's protocol must execute every
workload with as much parallelism as the primary's concurrency
control protocol. Otherwise lag can grow arbitrarily large.

Transaction- and page-granularity \ccc{} protocols are the current best
approaches. The former assume logical logs, and
the latter assume physical redo
logs~\cite{verbitski2018aurora,oracle2019adg,antonopoulos2019socrates}.
In transaction-granularity protocols, writes
conflict if they modify the same row, and the protocol serializes
transactions with conflicting
writes~\cite{hong2013kuafu,king1991remote,oracle2001txnscheduler,mysql8writeset}.
Page-granularity protocols serialize writes to each
page~\cite{verbitski2018aurora,oracle2019adg,antonopoulos2019socrates}. Both,
however, fail to guarantee bounded replication lag because for some workloads,
the primary executes with more parallelism than the backup.

\begin{figure}[tb]
  \centering
  \includegraphics[page=1,width=0.8\linewidth]{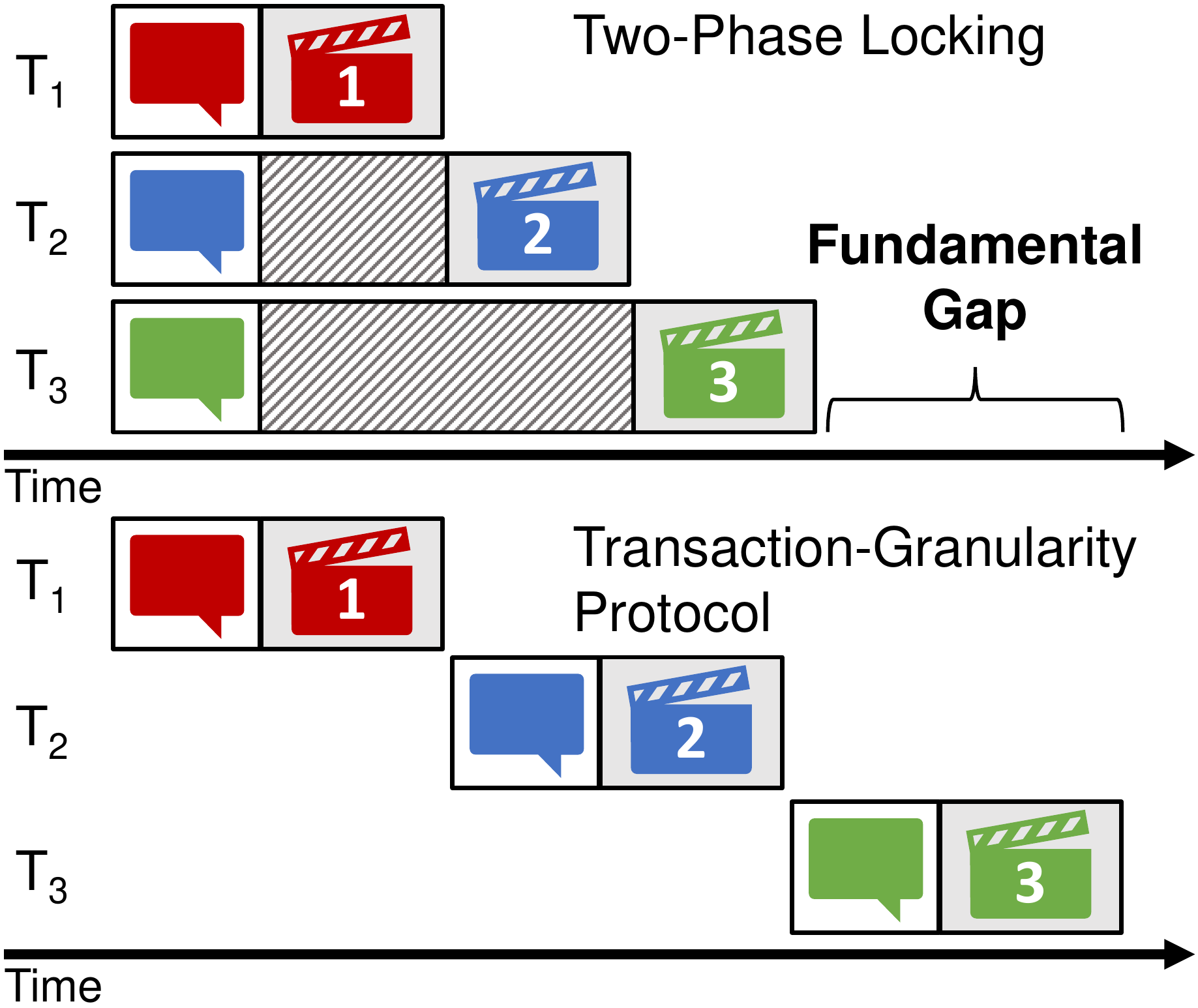}
  \caption{Primary (2PL) and backup (transaction-granularity protocol)
    executions when three users comment on the same video. Diagonal lines
    depict waiting for a lock.}
  \label{fig:interleavings}
\end{figure}

We can see how a transaction-granularity protocol can lag by returning to our motivating
example. Suppose Alice, Bob, and Charlie simultaneously comment on the same
video. Figure~\ref{fig:interleavings} shows a primary and backup's executions
of the six resultant operations. The primary uses two-phase
locking~\cite{bernstein-book} and stored procedures (i.e., no parsing and
planning); the backup implements a transaction-granularity
protocol~\cite{hong2013kuafu,king1991remote,oracle2001txnscheduler,mysql8writeset}.

On the primary, three threads insert rows in the comments table in parallel, but
updates to the video's comment counter are serialized by a row lock. On the
backup, however, the transaction-granularity protocol serially executes all of the operations. Even if the backup uses different \worker{}s to execute each
transaction, their execution is as at least as slow as that of one \worker{}. Thus, a fundamental
gap exists between the parallelism available to the primary and backup, which can cause arbitrarily long lag.

Page-granularity protocols have similar issues. Concurrency control protocols
using logical locking may allow concurrent transactions to update distinct
rows residing on the same physical page~\cite{oracle19,bernstein-book}.  Such concurrency control
can cause arbitrarily long replication lag with page-granularity
\ccc{} because writes that execute in parallel on the primary
are serialized on the backup. Thus, a fundamental gap again exists.

In the remainder of this section, we prove this problem is general to all
transaction-granularity protocols, and thus, no such
protocol guarantees bounded replication lag. We subsequently summarize the comparable
theorem for page-granularity protocols.

\subsection{Proof Of Unbounded Replication Lag}
\label{sec:scheduling:proof}

\noindentparagraph{System Model.} A \textit{database} $\mathcal{D}$ stores sets
$\mathcal{K}$ of \textit{keys} and $\mathcal{V}$ of \textit{values}. The
database's \textit{state} is a mapping from $\mathcal{K}$ to $\mathcal{V}$. A
\textit{transaction} $T$ is an ordered set of \textit{operations} (reads and
writes) on individual keys. For simplicity, we assume each value is
uniquely identifiable, so two identical transactions are, too.  We define
$R(T)$ and $W(T)$ as the sets of keys read and written by all operations in
$T$. We define transaction arrival times at the primary and backup as $a_p(T)$
and $a_b(T)$ respectively. Finally, we define the real time at which a
transaction is included in the primary and backup's state as $f_p(T)$ and
$f_b(T)$ respectively (as in Section~\ref{sec:bg:replication-lag}).

We assume a primary-backup system where both have $m$ \textit{cores}. In
isolation, the primary's cores each execute an operation in $e > 0$ time units.
The primary uses 2PL~\cite{bernstein-book}, so an operation may wait for a lock
if there is a concurrent operation on the same key. If there are multiple
conflicting operations, assume they are granted the lock in the order
requested.  To account for both eager and lazy \ccc{} protocols, assume the
backup's cores execute each operation in $0<d \leq e$ time units. We
assume $d \leq e$ because backups can avoid some processing done on the
primary, such as parsing and planning.  

When the primary finishes executing all operations in transaction $T$, it
records $T$'s writes in its log. $T_1 \prec T_2$ denotes that $T_1$ precedes
$T_2$ in the log. The log is then sent to the backup. For
simplicity, we assume this occurs instantaneously.

A \textit{workload} $W \in \mathcal{W}$ is a tuple $(\mathcal{T},
A_{\mathcal{T}})$ where $\mathcal{T}$ is a set of transactions and
$A_{\mathcal{T}}$ is a function from $\mathbb{R}$ to finite sets of transactions
$T \in \mathcal{T}$ representing the transaction arrival process at the primary.
$\mathcal{W}$ is the set of all definable workloads.

Transaction $T$'s \textit{replication lag} is given by $f_b(T) - f_p(T)$. A
\ccc{} protocol has finite replication lag for workload $W = (\mathcal{T},
A_{\mathcal{T}})$ if there
exists some finite $L$ such that for all $T \in \mathcal{T}$, $f_b(T) - f_p(T)
\leq L$, and it guarantees bounded replication lag if it has finite
replication lag for all $W \in \mathcal{W}$.

\noindentparagraph{Definitions \& Assumptions.} A \textit{transaction-granularity \ccc{}
  protocol} guarantees that for all pairs of transactions $T_1$ and
$T_2$, if $W(T_1) \cap W(T_2) \neq \emptyset$ and $T_1 \prec T_2$,
then all of $T_1$'s writes execute before any of $T_2$'s.

The proof below requires $m > \left\lceil \frac{e}{d} \right\rceil$, but this
assumption is reasonable in practice. Server CPUs commonly contain at least 64
physical cores~\cite{intel-servers}. Thus, the assumption is not satisfied only
if the backup executes operations more than $63$ times faster than the primary.
Stored procedures and sophisticated concurrency control
protocols~\cite{tu2013speedy, narula2014phase, wang2016mostly,
  larson2011concurrency} make such an advantage unlikely.

\begin{thm}
  Assume $m > \left\lceil \frac{e}{d} \right\rceil$. A primary-backup system
  that uses two-phase locking on its primary and a transaction-granularity
  \ccc{} protocol on its backup cannot guarantee bounded replication lag.
\end{thm}

\begin{proof}
  Assume we have a primary-backup system as described above, and assume to
  contradict that it guarantees bounded replication lag. Then there exists some
  $L$ such that for all $W \in \mathcal{W}$, the system executes all
  transactions in $W$ with replication lag $\leq L$.
  
  We now construct a workload $W \in \mathcal{W}$ that includes at least one
  transaction with replication lag greater than $L$. The workload comprises a
  set of $\left\lceil\frac{L}{nd - e}\right\rceil$ transactions of
  $n>\left\lceil\frac{e}{d}\right\rceil$ writes. Because $nd > e$, the number of
  transactions is well-defined. The first $n-1$ writes of each transaction modify unique keys, and
  the last updates key $k_0$. Define $A$ such that a new transaction arrives at the
  primary every $e$ time units, starting at $0$.

  Because the primary uses 2PL, it executes the first $n-1$ writes of each
  transaction in parallel but serializes their final updates to $k_0$. For convenience,
  we index the transactions in the order they appear in the primary's log.

  For the first set of $m$ transactions, $f_p(T_0) = ne$, $\ldots$, and
  $f_p(T_{m-1}) = (n+m-1)e$. Because $m>n$, the core that executed $T_0$ is
  free when $T_m$ arrives. Thus, $T_m$ finishes $e$ time units after $T_{m-1}$.
  In general, we see $f_p(T_i) = (n+i)e$.

  The backup uses a transaction-granularity protocol, so it serially executes
  all writes in the workload. Thus, the backup finishes executing $T_0$ at
  $n(e+d)$. By construction, $nd \geq e$, so $f_p(T_1) \leq f_b(T_0)$.
  Thus, the backup immediately starts executing $T_1$ after $T_0$. The same is
  true for all subsequent transactions. In general, we see $f_b(T_i) = ne + (i+1)nd$.

  Thus, in general, $f_b(T_i) - f_p(T_i) = ne + (i+1)nd - (n+i)e = i(nd-e) + nd$. 
  For the final transaction $T$ in the workload, $i = \left\lceil\frac{L}{nd - e}\right\rceil$, and thus
  $f_b(T) - f_p(T) = \left\lceil\frac{L}{nd - e}\right\rceil(nd-e) + nd \ge \frac{L}{nd - e}(nd-e) + nd > \frac{L}{nd - e}(nd-e)$.
  Equivalently, $f_b(T) - f_p(T) > L$, a contradiction.
\end{proof}

The result above shows that if the primary has sufficient cores, then a
transaction-granularity protocol cannot guarantee bounded replication lag. To simplify our formalism, the proof assumes the primary uses 2PL and serializable isolation~\cite{berenson1995critique,papadimitriou1979serializability}. We note three important extensions: First, because it can only accelerate the primary but not the backup, the theorem applies even if the primary uses weaker isolation~\cite{adya1999consistency}. 

Second, a similar result can be derived for some optimistic protocols~\cite{bernstein-book,lim2017cicada}. For example, a similar execution to the one in Figure~\ref{fig:interleavings} is possible with multi-version timestamp ordering (MVTSO)~\cite{bernstein-book}. Using MVTSO, the three transactions still insert comments in parallel. If they then read the comment counter, write its new value, and perform validation serially in timestamp order, all three transactions will commit, and a fundamental gap will again exist. We leave the generalization of our formal framework to optimistic concurrency control to future work.

Third, because the above proof only assumes $0 < d \leq e$, it also applies even if there is primary-specific processing, such as parsing and planning. This additional processing can be accounted for by increasing $e$, and the theorem holds as long as $m > \left\lceil \frac{e}{d} \right\rceil$ remains true. If it does not, then the primary-backup system may be able to guarantee bounded replication lag, but only because the bottlenecks on the primary make it easy for an inefficient \ccc{} protocol to keep up.

More generally, the proofs above and below will not apply if bottlenecks prevent the primary from executing quickly enough or with enough parallelism, such as those in logging, persistence, or log transfer. Despite these cases, however, solving replication lag remains urgent: First, there are many cases where these are not bottlenecks. For parsing and planning, many deployments use stored procedures. For persistence, mechanisms such early lock release
~\cite{johnson2010aether,gawlick1985varieties,dewitt1984implementation} and epoch-based group commit~\cite{tu2013speedy,chandramouli2018faster} help decouple transaction throughput from I/O latency. Second, we expect advances in research and technology, such as non-volatile memory~\cite{intelOptane}, to eventually remove these bottlenecks.

\subsubsection{Page-granularity Protocols Cannot Keep Up.}
\label{sec:scheduling:page-granularity}
We use a structurally similar proof for page-granularity protocols, but we omit it here due to space constraints. Similar to the one above, the proof shows that if the primary has sufficient cores, can fit enough rows on each page, uses 2PL~\cite{bernstein-book}, and guarantees serializable isolation~\cite{papadimitriou1979serializability}, then a page-granularity protocol cannot guarantee bounded replication lag. The assumption about the number of cores is identical to the one assumed above.

Unlike the previous proof, we additionally assume that the number of rows that can fit on a page is greater than $\left\lceil \frac{e}{d} \right\rceil$. But like the assumption about cores, this assumption is reasonable in practice. With a typical cache line size of
\SI{64}{\byte}---more than enough to store a row with two integer columns---and page size of \SI{4}{\kibi\byte}, 64 rows can be stored on the same
page, each on a different cache line. Thus again the assumption will not hold
only if the backup can execute operations more than $63$ times
faster than the primary.

%% file: design.tex
\section{Design}
\label{sec:design}

\sys{} achieves two competing goals: it ensures bounded replication lag and
guarantees monotonic prefix consistency for read-only transactions. To
accomplish this, \sys{} comprises three components: a \scheduler{}, a set of \worker{}s, and a \snapshotter{}. The \scheduler{} and \worker{}s ensure bounded replication lag by 
executing writes at a sufficiently fine granularity, and
the \snapshotter{} guarantees read-only transactions only see changes that are
valid under monotonic prefix consistency. Together they implement \sys{}'s
row-granularity \ccc{} protocol.

As shown in Section~\ref{sec:scheduling}, transaction- and page-granularity
protocols fail to provide bounded replication lag because they cannot always
execute with the same parallelism as the primary. The primary executes writes to
different rows in parallel. Thus, to provide bounded lag, \sys{}'s
\worker{}s execute writes at row granularity.\footnote{Some concurrency control protocols allow   two threads to update the same row's cells in parallel~\cite{huang2020stov2}. For ease of exposition, we assume they cannot, but rows are not fundamental to
  our design---\sys{} could be adapted for finer granularities.}

Unconstrained row-granularity execution, however, can lead to permanent violations of monotonic prefix
consistency because conflicting writes may execute in the wrong order. For instance, suppose two transactions $T$ and $U$ each update
rows $x$ and $y$. If different \worker{}s execute the resultant writes (denoted $w_T[x]$, $w_T[y]$,
$w_U[x]$, and $w_U[y]$), $w_T[x]$ may finish before $w_U[x]$ and $w_U[y]$
before $w_T[y]$. If there are no further writes to these rows, the backup will
forever reflect $w_U[x]$ and $w_T[y]$, violating transactional atomicity and thus MPC.
\sys{}'s \scheduler{} helps avoid permanent consistency violations by constraining the \worker{}s execution. These constraints ensure
that writes to each row are applied in the same order as on the primary. Thus,
each row reflects monotonically increasing prefixes of the log.

But per-row monotonicity is insufficient to guarantee global monotonic prefix
consistency. In the example above, a write to a third row $z$ from a third
transaction $V$ may be scheduled after $T$ and $U$ but applied first. If this occurs, then a
read-only transaction of rows $x$, $y$, and $z$ would violate monotonic prefix
consistency. 
Instead, \sys{}'s \snapshotter{} uses a set of three progressing database
snapshots to allow uninterrupted execution of non-conflicting writes while
guaranteeing MPC.

\begin{figure}[t]
  \centering
  \includegraphics[page=8,width=0.85\linewidth]{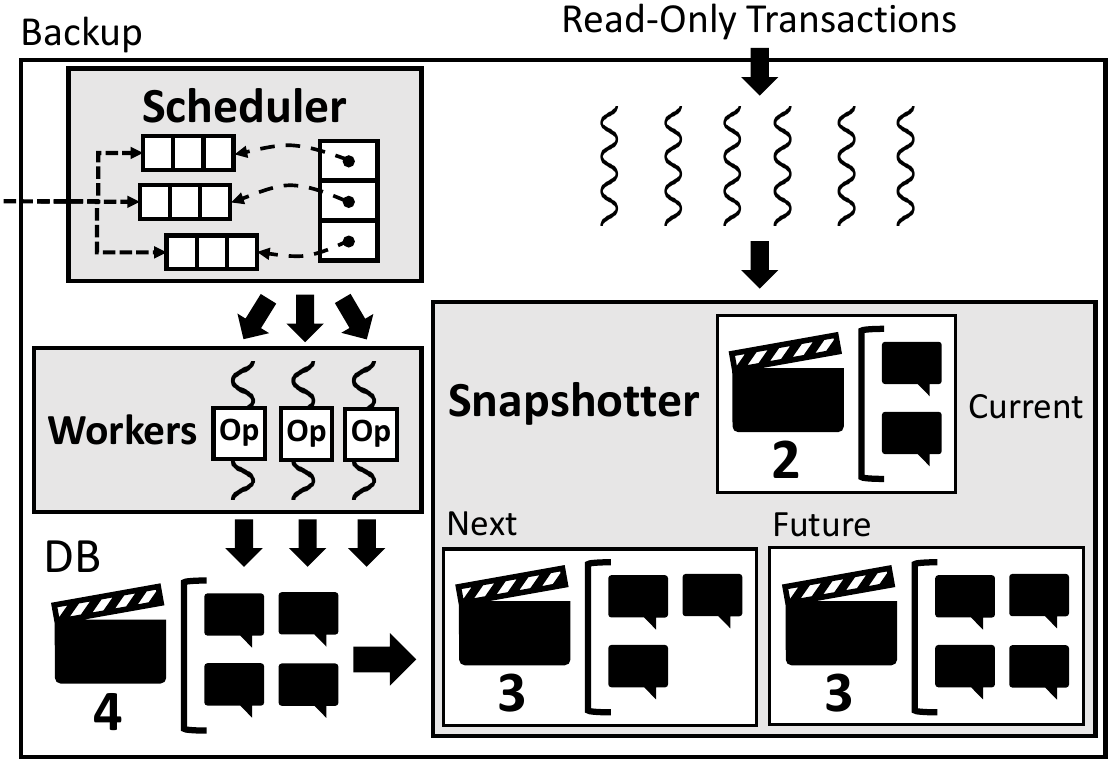}
  \caption{\sys{}'s \scheduler{}, \worker{}s, and \snapshotter{}.}
  \label{fig:system-overview}
\end{figure}

\begin{figure*}[t]
  \centering
  \includegraphics[page=11,width=0.75\linewidth]{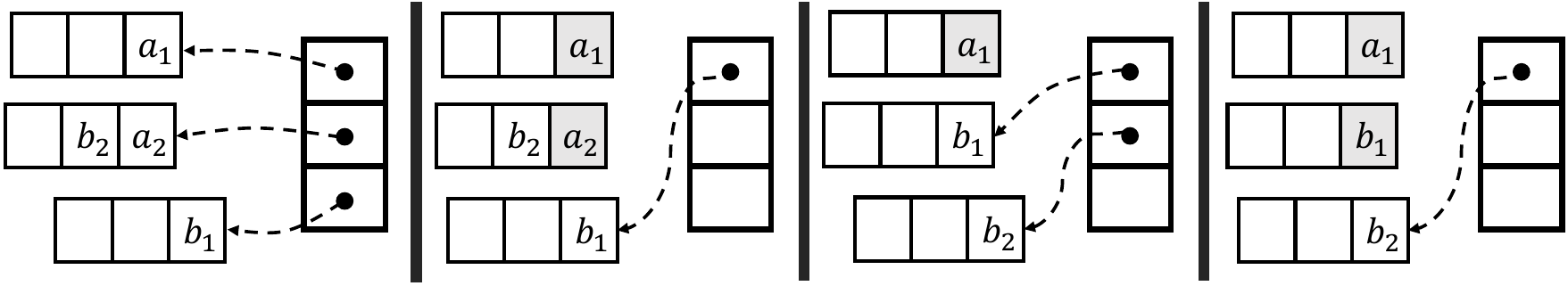}
  \caption{Left to right shows the \scheduler{}'s queues as two \worker{}s execute four writes. Grey writes are being executed.}
  \label{fig:scheduler}
\end{figure*}

Figure~\ref{fig:system-overview} shows \sys{}'s design. The \scheduler{}
orders writes and schedules them for execution by the \worker{}s. 
The \snapshotter{} exposes a monotonic-prefix
consistent view of the database to read-only transactions, which are executed
by a separate set of threads. We now describe \sys{}'s components
in turn.

\subsection{Row-Granularity Scheduling \& Execution}
\label{sec:design:scheduler}

As described in Section~\ref{sec:bg:replication}, the backup continuously
receives a log of operations from the primary, including the rows written by each
operation and metadata to delimit transactions. To guarantee bounded
replication lag, \sys{}'s \worker{}s must execute individual row writes while obeying the constraints specified by the \scheduler{}. To avoid permanent consistency violations, the
\scheduler{} logically constructs a FIFO queue for each row whose order reflects the order of the row's writes in the log.

As the \scheduler{} processes writes, it assigns each a sequence number, which
reflects the write's position in the log. The \scheduler{}
then enqueues the write in the appropriate FIFO queue.



A write is \textit{safe} to execute when it reaches the head of its FIFO queue and the prior
head has finished executing. This assumes the backup receives the log of each row's writes and the \scheduler{} processes
them in order.
(The log shipping subsystems in many commercial databases
satisfy this assumption~\cite{mysql,postgres,oracle19}.) Given this, the
\scheduler{} is assured that when it processes a write, all conflicting writes
that precede it in the log either are already in the queue or are executing.

To keep replication lag small, the \scheduler{} ensures \worker{}s execute safe writes promptly. To do so, the \scheduler{} uses a FIFO queue to order the
queues described above. To avoid ambiguity, we refer to a \scheduler{} queue and per-row queues. Thus, a \worker{} chooses the next write for
execution by first removing the per-row queue at the head of the \scheduler{} queue and then executing the write at its head.
When the \worker{} finishes executing the write, the per-row queue is
reinserted into the \scheduler{} queue.


To demonstrate how \sys{}'s \scheduler{} and \worker{}s operate, we return to our motivating
example when Alice and Bob concurrently comment on the same video. Assume Alice's
transaction $A$ commits first. It performs two operations: operation $a_1$
inserts one comment row, and $a_2$ increments the video's comment counter. Bob's
transaction $B$ performs comparable operations $b_1$ and $b_2$.

Figure~\ref{fig:scheduler} shows, from left to right, how two \worker{}s execute the four writes. The first panel shows the initial data
structures after processing all operations. $a_1$ and $a_2$ then begin
executing (second panel). Assume $a_2$ finishes
before $a_1$. When it does, its corresponding queue is first reinserted at the tail of the \scheduler{} queue (third panel). Finally, $b_1$ starts
executing (fourth panel). This process continues until all writes finish.

We prove below that \sys{}'s
row-granularity execution never imposes more constraints on the backup than any concurrency control protocol imposes on the primary.
Thus, \sys{} can always execute with equal or greater parallelism as the primary and keep up.

\subsubsection{\sys{}'s Execution Can Keep Up.}
\label{sec:design:scheduler-proof}

Both the primary's concurrency control and the backup's \ccc{} protocols can be viewed as functions from a set of logs to a set of sets of \textit{execution schedules.} Given a log, the primary's threads and the backup's workers execute its writes according to one of the schedules in its image. We say a schedule is \textit{valid} if the schedule of writes produces an equivalent database state as serially executing the writes in the log. In the remainder of this section, we only consider the set of \textit{valid protocols}, those whose images contain only sets of valid schedules, and denote it as $\mathcal{A}$. Note that a primary's concurrency control protocol is always in $\mathcal{A}$ because the primary's durability guarantees that its log, when serially executed, reproduces its state.

Let $w_T[x]$ denote a write to row $x$ by transaction $T$. As before, $T \prec U$ denotes transaction $T$ precedes transaction $U$ in the log, and in a slight abuse of notation, let $w_T[x] \prec w_U[y]$ denote write $w_T[x]$ precedes write $w_U[y]$ in the log.  Similarly, $w_T[x] < w_U[y]$ denotes $w_T[x]$ precedes $w_U[y]$ in an execution schedule. A \textit{row-granularity protocol} guarantees that for all logs and all pairs of writes $w_T[x]$ and $w_U[y]$, if $x = y$ and $T \prec U$, then $w_T[x] < w_U[y]$ in all of its possible schedules.

\begin{thm}
  \label{thm:scheduler-keeps-up}
  Let $R \in \mathcal{A}$ be a row-granularity protocol. Given a log, there does not exist a valid protocol $P \in \mathcal{A}$ that imposes fewer constraints on its corresponding set of execution schedules than $R$.
\end{thm}

\begin{proof}[Proof Sketch]
  Given a log and two transactions $T$ and $U$ such that $T \prec U$, $R$ imposes one constraint on the possible executions of their writes: if $w_T[x]$ conflicts with $w_U[x]$, then $w_T[x] < w_U[x]$.
  
  Assume to contradict there is a valid protocol $P$ that does not impose the above constraint. Then in one of the resulting schedules, $w_U[x] < w_T[x]$. A serial execution of the log, however, always executes $w_T[x]$ before $w_U[x]$ because $T \prec U$, and thus $w_T[x] \prec w_U[x]$. As a result, this execution is not equivalent to the serial execution of the log, contradicting that $P$ is valid.
\end{proof}

The consequence of this proof is that all valid concurrency control protocols
on the primary, regardless of isolation level, must at a minimum impose the
constraints imposed by a row-granularity protocol on the backup. If a
backup employs such a protocol, then regardless of how much parallelism is exploited by the primary's concurrency control during its execution, an execution with an equal degree of parallelism is available to the backup. Thus, the proof shows that row-granularity execution never limits the backup's ability to match the primary's parallelism.

But practical considerations, such as poorly designed concurrency mechanisms or bottlenecks in the \scheduler{}, may prevent the backup's \worker{}s from always keeping up with the primary despite employing a row-granularity protocol. To avoid complicating the formalism above, we thus do not prove that any specific implementation guarantees bounded replication lag. Instead, our experimental evaluations in Sections~\ref{sec:myrocks-evaluation} and~\ref{sec:cicada:evaluation} verify that the principles learned here translate into bounded lag in practice.

\subsection{\Snapshotter{} \& Read-Only Transactions}
\label{sec:design:snapshot}

\sys{}'s \snapshotter{} uses database snapshots to guarantee monotonic prefix consistency without blocking \worker{}s. Table~\ref{tbl:logical-storage-api} shows the logical storage API needed to implement it. While this API is not
backward-compatible with the storage engines in some commercial
databases~\cite{postgres2020storage,innodb,rocksDB,leveldb}, it can be implemented efficiently in many modern databases where \worker{}s can explicitly assign timestamps to their writes~\cite{levandoski2011deuteronomy,lim2017cicada,diaconu2013hekaton,kim2016ermia}. We elaborate on how this difference affects the \snapshotter{}s of \sysmyrocks{} and \syscicada{} in Sections~\ref{sec:myrocks-impl:snapshot} and~\ref{sec:cicada:impl}, respectively.


\begin{table}[t]
  \small
  \centering
  \setlength{\tabcolsep}{0.5\tabcolsep}
  \begin{tabular}{@{}l l@{}}
    \toprule
    \textbf{Function} & \textbf{Description}
    \\ \midrule
    $\textsc{NewSnapshot}(\mathcal{D}) \to \mathcal{S}$ & Create empty
                                                          snapshot.
    \\ \addlinespace[1mm]
    $\textsc{Merge}(\mathcal{S}_1, \mathcal{S}_2) \to \mathcal{S}_3$ & Merge
    $\mathcal{S}_1$ and $\mathcal{S}_2$; $\mathcal{S}_3$ reflects all \\
    & writes to both, in order.
    \\ \addlinespace[1mm]
    $\textsc{Read}(\mathcal{S}, r) \to v$ & Read value from snapshot.
    \\ \addlinespace[1mm]
    $\textsc{Insert}(\mathcal{S}, r, v)$ & Add
    row-value pair to snapshot.
    \\ \addlinespace[1mm]
    $\textsc{Update}(\mathcal{S}, r, v)$ & Update
    row-value pair.
    \\ \addlinespace[1mm]
    $\textsc{Delete}(\mathcal{S}, r, v)$ & Delete row-value pair.
    \\ \bottomrule
  \end{tabular}
  \vspace{5mm}
  \caption{Logical storage interface for \sys{}. Snapshots are
    sequences of writes. Empty snapshots are created, and two snapshots can be
    merged. \Worker{}s add writes to
    snapshots.}
  \label{tbl:logical-storage-api}
\end{table}

As Table~\ref{tbl:logical-storage-api} shows, new snapshots are created from the database. Logically, a snapshot is a sequence of writes, so it is initially empty. Writes directly modify a snapshot. Two
snapshots $\mathcal{S}_1$ and $\mathcal{S}_2$ can be merged to produce a
third $\mathcal{S}_3$ that reflects the writes applied to both, with all
writes in $\mathcal{S}_1$ ordered before those in $\mathcal{S}_2$. Finally,
the latest version of a row's value can be read from a snapshot.

The \snapshotter{} uses this API to maintain three database
snapshots, logically representing the current, next, and future. The current
snapshot is initially empty, always prefix-complete, and serves read-only
transactions. The next and future snapshots are initially empty. \Worker{}s only modify the next and future snapshots.

Figure~\ref{fig:snapshots} illustrates how the \snapshotter{} incorporates a
write into the current snapshot while maintaining monotonic prefix
consistency. \sys{}'s \snapshotter{} uses two sequence numbers to delimit
the three snapshots. The current snapshot includes all writes up to sequence
number $c$. All writes with
sequence numbers between $c$ and $n$ (inclusive) update the next snapshot; all
writes with sequence numbers greater than $n$ update the future snapshot.


When all writes with sequence numbers between $c$ and $n$ finish executing, the
current and next snapshots together form a new, prefix-complete snapshot. The
\snapshotter{} then merges them, and the result replaces the current. At the
same time, it performs four additional operations: $c$ is updated to reflect
the new current snapshot; $n$ is advanced; the next snapshot is replaced with
the future snapshot; and a new future snapshot is created.



To satisfy monotonic prefix consistency, the \snapshotter{} always aligns $n$
with a transaction boundary.
Thus, the
next snapshot always reflects a set of complete transactions before being merged.


Because they execute against different snapshots, \worker{}s and
read-only transactions execute entirely in parallel. But to guarantee
bounded replication lag, \worker{}s must be given higher scheduling
priority than read-only transactions threads. To avoid starvation, we assume
they execute on separate cores (beyond the $m$ assumed in
Section~\ref{sec:design:scheduler-proof}), or if they execute on the same cores, there are enough spare CPU cycles
to process all read-only transactions.


%% file: myrocks_impl.tex
\section{\sysmyrocks{} Implementation}
\label{sec:myrocks-impl}

\sysmyrocks{} was developed to solve replication lag at \fb{}, so backward compatibility and ease of deployment were primary concerns. To remain backward-compatible with \myrocks{} (a fork MySQL that uses RocksDB as its
storage engine~\cite{mysql,myRocks,rocksDB}), \sysmyrocks{} imposes some additional constraints on its execution, beyond those discussed
in Section~\ref{sec:design}. In this section, we describe the
implementation, highlighting how it leverages MyRocks's existing
features~\cite{mysql,myRocks} and differs from our design.

\subsection{Scheduling \& Execution}
\label{sec:myrocks-impl:scheduler}

\begin{figure}[t]
  \centering
  \includegraphics[page=3,width=0.85\linewidth]{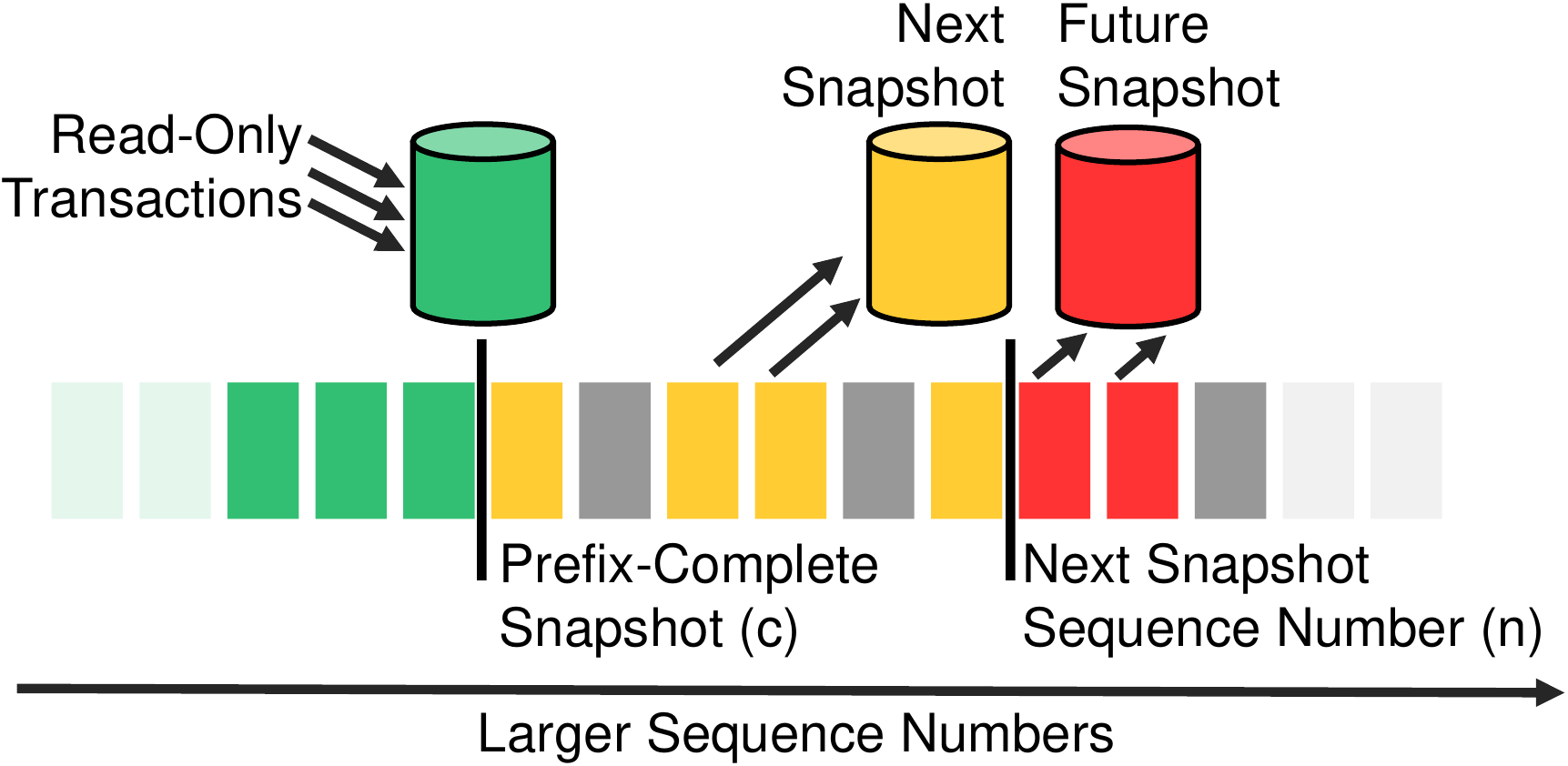}
  \caption{\sys{}'s \snapshotter{}. Writes in color (grey) have finished
    (not finished) executing.}
  \label{fig:snapshots}
\end{figure}

The \scheduler{} leverages MyRocks's row-based-logging
subsystem~\cite{mysql,myRocks}, so \sysmyrocks{}
does not require any changes to the primary or its log. For each
operation, the log includes the set of written keys and corresponding
values.

\myrocks{}'s row-based logging subsystem assumes that all of a transaction's writes are executed by the same \worker{}. To keep its changeset small (630 lines of C++ code), \sysmyrocks{} also enforces this constraint, in addition to those described in Section~\ref{sec:design:scheduler}.

To do so, the \scheduler{} builds a linked list of each transaction's writes. After the \scheduler{} adds all of a transaction's writes to
their per-row queues, it puts the transaction's first write in the \scheduler{} queue. After a \worker{} dequeues it, the \worker{} first waits
until the write reaches the head of its per-row queue (i.e., it is safe to execute), then executes it, and finally follows the pointer to the next write in the transaction, repeating this for any remaining writes. Once finished, the \worker{} takes the next write from the \scheduler{} queue.

The one-thread-per-transaction execution model and having \worker{}s pick up transactions in commit order greatly simplified the implementation, but this simplicity is not free. \sysmyrocks{}'s
implementation is more constrained than our design and thus executes some workloads with less parallelism. Despite these extra constraints, however, our evaluation demonstrates that \sysmyrocks{} can keep up with its primary.

\subsection{\Snapshotter{} \& Read-Only Transactions}
\label{sec:myrocks-impl:snapshot}

The storage engines in some widely deployed databases~\cite{postgres2020storage,innodb,rocksDB,leveldb}, including \myrocks{}, cannot easily implement the API in Table~\ref{tbl:logical-storage-api}. Unfortunately,
they are also complex, comprising tens to
hundreds of thousands of lines of code~\cite{rocksDB,leveldb}. To again keep \sysmyrocks{}'s changeset small, we opted to implement its \snapshotter{} without requiring changes to RocksDB's interface.

In \myrocks{}, snapshots are read-only and can only be taken of the database's current state. Neither \worker{}s nor the \snapshotter{} have fine-grained control over which writes are included in a snapshot (e.g., by taking a snapshot as of some specified timestamp or version number). As a result, the \snapshotter{} must impose some additional constraints on the \worker{}s to ensure the entire database is prefix-consistent when taking a new snapshot. 



Instead of three snapshots, \sysmyrocks{}'s \snapshotter{} logically maintains two: It uses a current snapshot, which is always prefix consistent and used to serve read-only transactions. The next snapshot, however, is replaced by the database. $c$ still tracks the writes included in the current snapshot.

To merge the current and next snapshots, the \snapshotter{} performs the following: First, it chooses $n$, the sequence number of the last write to be included in the next snapshot. To ensure the merged snapshot is prefix-consistent, choosing $n$ also blocks \worker{}s from executing writes with sequence numbers greater than $n$ until after the
snapshot is taken. (If the storage engine supports transactions, as in RocksDB~\cite{rocksDB}, the \worker{}s only need delay committing these writes.) Second, after all writes with sequence numbers between $c$ and $n$ (inclusive) execute, the \snapshotter{} takes a new snapshot of the database and replaces the current one. $c$ and $n$ then advance as described previously. Advancing $n$ allows blocked \worker{}s to proceed with their writes.

If taking a snapshot is \fbok{computationally expensive}, the blocking above may
lead to spikes in replication lag. To combat this, our implementation allows
database administrators to tune the approximate snapshot frequency, $I$, in milliseconds. Because the
storage engine may not be prefix consistent exactly every $I$ milliseconds, the \snapshotter{} advances $n$ by an estimate of the number of writes \worker{}s will execute in the next $I$ milliseconds.


By tuning $I$, administrators can ensure replication lag returns to satisfactory
levels between snapshots provided lag can decrease between
snapshots. This implies the backup must execute each write marginally faster
than the primary. But we found this assumption reasonable in practice, and our evaluation further supports this.


%% file: myrocks_evaluation.tex
\section{\sysmyrocks{} Evaluation}
\label{sec:myrocks-evaluation}

Our evaluation explores the following questions:
\begin{enumerate}
\item Does \sysmyrocks{} help engineers avoid potential disasters caused by
  optimizations of realistic workloads? (\S\ref{sec:evaluation:tpcc})

\item Does \sysmyrocks{} always keep up with the primary? (\S\ref{sec:evaluation:throughput})

\item Does \sysmyrocks{} guarantee MPC for read-only transactions without causing
  unbounded replication lag? (\S\ref{sec:evaluation:reads})
\end{enumerate}

\noindentparagraph{Experimental Setup.} All experiments ran on the CloudLab Wisconsin
platform~\cite{duplyakin2019cloudlab} with three servers located in one
datacenter: one for load generation, the primary, and the
backup. Round-trip times between machines were less than
\SI{100}{\micro\second}. Each machine had two 2.20 GHz Intel Xeon processors
with ten cores each, hyper-threading enabled, \SI{192}{GB} of RAM, a \SI{10}{Gb}
NIC, and a \SI{480}{GB} SSD.

All results were from 120-second trials. We omit all data from the first and
last 15 seconds of each trial to avoid experimental artifacts. Unless otherwise specified, we
ran each experiment 5 times and report the median result. For each experiment,
we generated load with a fixed number of closed-loop clients. The number of
clients and \worker{}s were set to maximize the primary and backup's
throughput, respectively. The number of \worker{}s never exceeded the number of primary threads.

Where noted, an additional set of closed-loop clients sent read-only
transactions to the backup. The backup's \worker{}s and read-only threads were
pinned to separate cores.

The log and \myrocks{}'s state were asynchronously flushed to disk on both the
primary and backup~\cite{rocksDB}. In all experiments, disk, memory, and network
were not bottlenecks. For all implementations, we used read-free replication and
disabled \myrocks{}'s 2PL on the backup~\cite{rocksDB,myRocks} since the \scheduler{} already prevents conflicting writes from executing
concurrently. Finally, to stress the backup, the primary used read
committed isolation~\cite{berenson1995critique}.

\noindentparagraph{Workloads.} We use three workloads. The first is TPC-C~\cite{tpcc}, a standard OLTP benchmark simulating an
order-entry application.
The other two, insert-only and adversarial, are synthetic. In each, the
database contains one table with two integer columns, a primary key and its associated value.

Each transaction in the insert-only workload comprises a variable number of unique inserts. Each transaction in the adversarial workload comprises a variable number of unique inserts and one update; the updates
in all transactions set the same row's value to a random integer, so all transactions
conflict.

\noindentparagraph{Baselines.} \kuafu{}~\cite{hong2013kuafu}, a
state-of-the-art, transaction-granularity \ccc{} protocol, is our baseline.
\kuafu{}'s protocol is nearly identical to MySQL 8's
write-set-based parallel replication~\cite{mysql8writeset} and is strictly better than the database-granularity and epoch-based protocols used in earlier versions of MySQL~\cite{mysql,mysqlgroupcommit} and its variants~\cite{mariadb10}.
We re-implemented \kuafu{} in \myrocks{}.

\subsection{\sysmyrocks{} Prevents Potential Disasters}
\label{sec:evaluation:tpcc}

\begin{figure}[t]
  \includegraphics[width=0.85\linewidth]{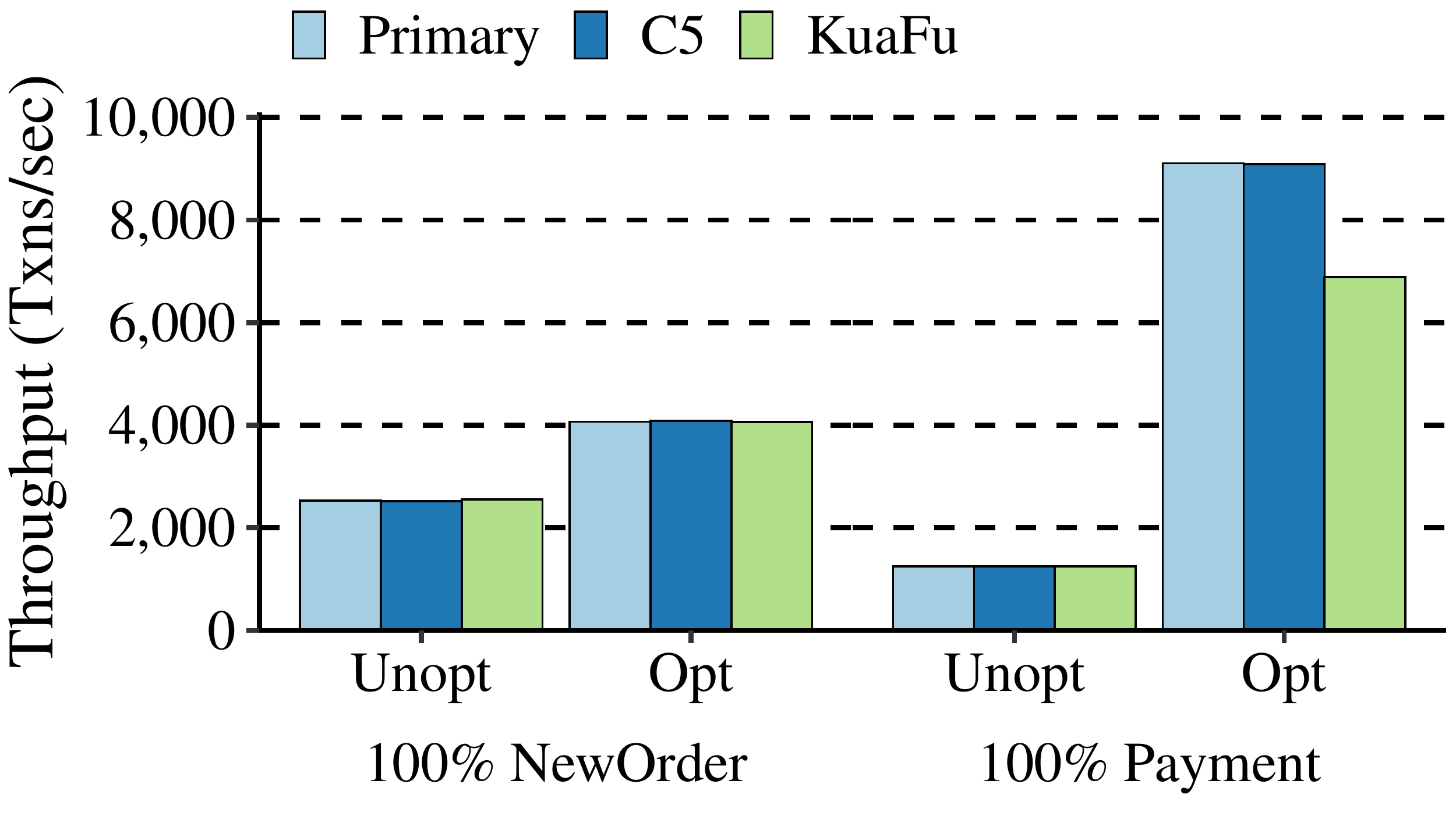}
  \caption{Throughput with 100\% NewOrder and Payment workloads
    before and after optimization. \kuafu{} lags after some optimizations while
    \sysmyrocks{} always keeps up.}
  \label{fig:tpcc}
\end{figure}

Because software and hardware improvements have accelerated the primary's
processing, primary-backup systems using transaction-granularity protocols are
brittle. Simple changes to a workload may cause
unbounded replication lag. To demonstrate this problem, we use TPC-C~\cite{tpcc}. While \kuafu{}~\cite{hong2013kuafu} keeps up on the standard benchmark workload, simple optimizations and non-standard transaction mixes
cause unbounded replication lag with the same
protocol~\cite{hong2013kuafu}.
We discuss the optimizations and our results in turn.

We optimize two of TPC-C's transactions: the NewOrder and Payment
transactions~\cite{tpcc}. In both cases, we defer higher-contention operations as much as possible while preserving application semantics. (Similar optimizations were observed in prior work~\cite{yan2016leveraging}.) In the NewOrder
transaction, the highest contention write is the increment of the district's
next order ID. In the Payment transaction, it is the update to the
warehouse's balance~\cite{tpcc}. Deferring these writes allows more parallelism on the primary.


Figure~\ref{fig:tpcc} shows the primary and backup's
throughput for \SI{100}{\percent} NewOrder and Payment workloads before and after optimization.
For the NewOrder workload, the optimization
increases the primary's throughput from \SI{2527}{} to
\SI{4067}{transactions\per\second}. For the Payment workload, the primary's throughput
increases by over \SI{700}{\percent} from \SI{1249}{} to
\SI{9105}{transactions\per\second}.

\kuafu{} keeps up with the optimized NewOrder workload. Data
dependencies between operations limit how late the district
row write can be deferred within each transaction and in turn, limits the primary's parallelism. But
\kuafu{} cannot keep up on the optimized Payment workload; its throughput
peaks at \SI{6889}{transactions\per\second}. Conversely, \sysmyrocks{} keeps
up.

If the optimization to the Payment transaction were made to a production
workload, \kuafu{} would cause significant replication lag, with 2216
transactions queuing at the backup every second. This rate is larger than the
one that induced replication lag of nearly 2 hours in production at \fbprom{} (discussed further in Section~\ref{sec:deployment}). On the other
hand, \sysmyrocks{} thwarts such a disaster.

\subsection{\sysmyrocks{} Always Keeps Up}
\label{sec:evaluation:throughput}

To validate that \sysmyrocks{} always keeps up, we measured the primary and \sysmyrocks{}'s throughput using the insert-only and adversarial workloads. The two
are on opposite ends of the contention spectrum: no transactions
conflict in the former, while all do in the latter. We also present \kuafu{}'s results.


\begin{figure}[t]
  \centering
  \includegraphics[width=0.85\linewidth]{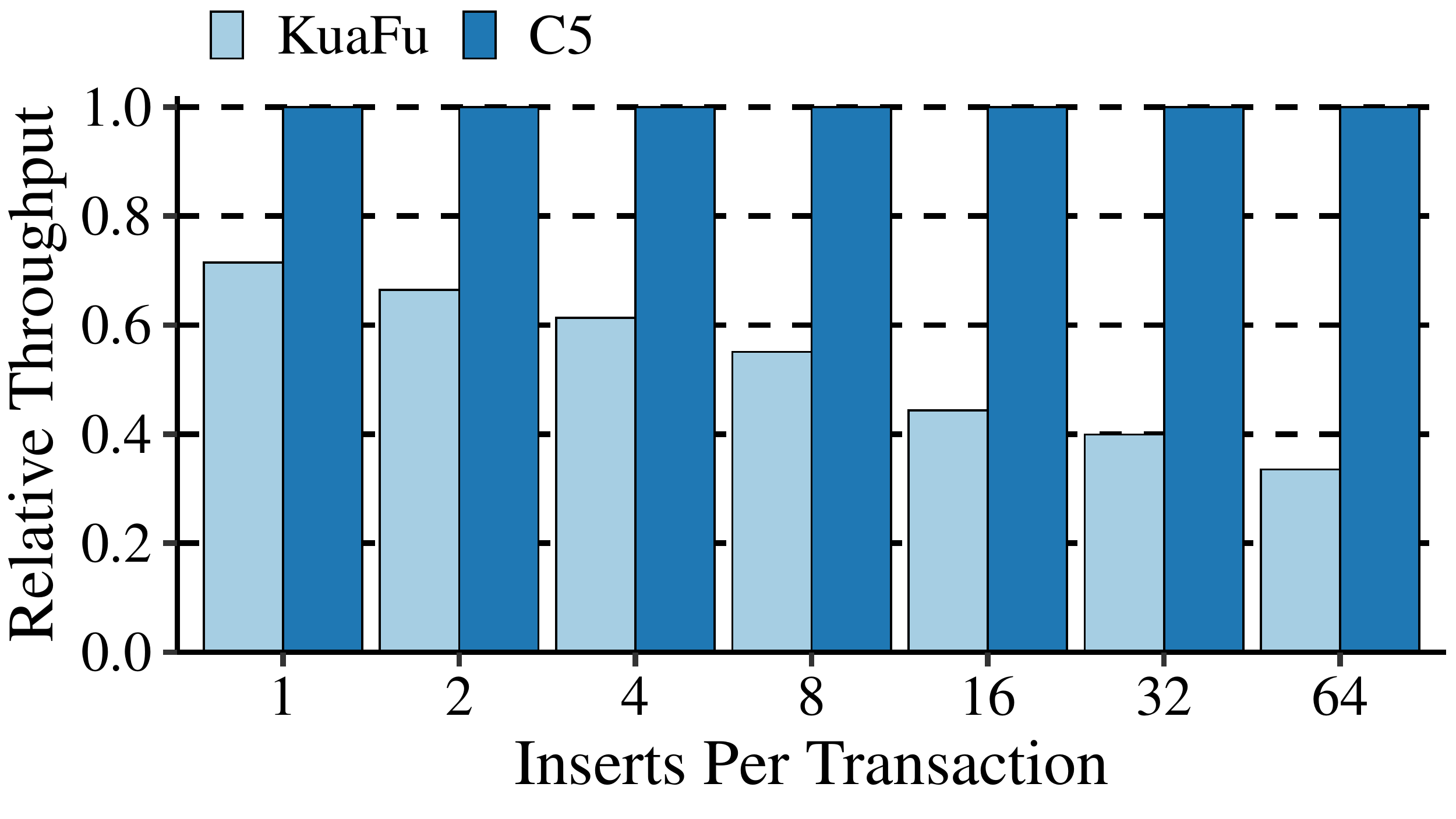}
  \caption{Backup's throughput relative to the primary's for the adversarial
    workload. As the number of inserts per transaction increases,
    transaction-granularity protocols lag more while \sysmyrocks{} always keeps up.}
  \label{fig:adversarial-workload}
\end{figure}

Because all transactions are non-conflicting, the insert-only workload stresses the primary's concurrency control and backup's \ccc{}.
Here, \myrocks{}'s throughput is about \SI{40500}{transactions\per\second}.
\sysmyrocks{} keeps up with the primary, indicating that its scheduling mechanisms
have sufficiently low overhead. As expected, \kuafu{} also keeps up. With both protocols, incoming writes can be executed immediately.

To verify \sysmyrocks{}'s \scheduler{} is not a bottleneck, we ran the same
experiment offline. We loaded the primary with inserts as above but
delayed replication. Once all writes finished and the resultant log
was transferred, we enabled \sysmyrocks{}'s \scheduler{} and \worker{}s.
We used sufficient \worker{}s, so the \scheduler{} was the bottleneck. \sysmyrocks{}'s
\scheduler{} processed \SI{95683}{transactions\per\second}, more than double \myrocks{}'s throughput.

Figure~\ref{fig:adversarial-workload} shows each implementation's performance
on the adversarial workload. We plot the backup's throughput relative to the
primary's as we vary the number of non-conflicting inserts per transaction from
1 to 64. Every transaction updates the same row. Despite the high contention,
the primary, using 2PL~\cite{bernstein-book}, executes the non-conflicting inserts
that precede the conflicting update in parallel. Because all transactions
conflict, \kuafu{} serializes them. Thus, the primary's advantage over \kuafu{}
increases with the the number of inserts. \kuafu{}'s
throughput drops from \SI{70}{\percent} to just \SI{38}{\percent} of the
primary's. On the other hand, \sysmyrocks{} executes the
non-conflicting inserts in parallel so always keeps up.

\subsection{\sys{} Serves Reads With Bounded Lag}
\label{sec:evaluation:reads}

\sysmyrocks{}'s implementation blocks writes from committing while taking a
snapshot. But it also exposes a parameter to tune the frequency of snapshots.
With periodic snapshots, \sysmyrocks{} serves read-only transactions in parallel
with writes, and thus, steady-state replication lag remains bounded despite
additional load from read-only clients. To validate these claims, we measure
replication lag as we increase the read-only load on \sysmyrocks{}.

\begin{figure}[t]
  \centering
  \includegraphics[width=0.85\linewidth]{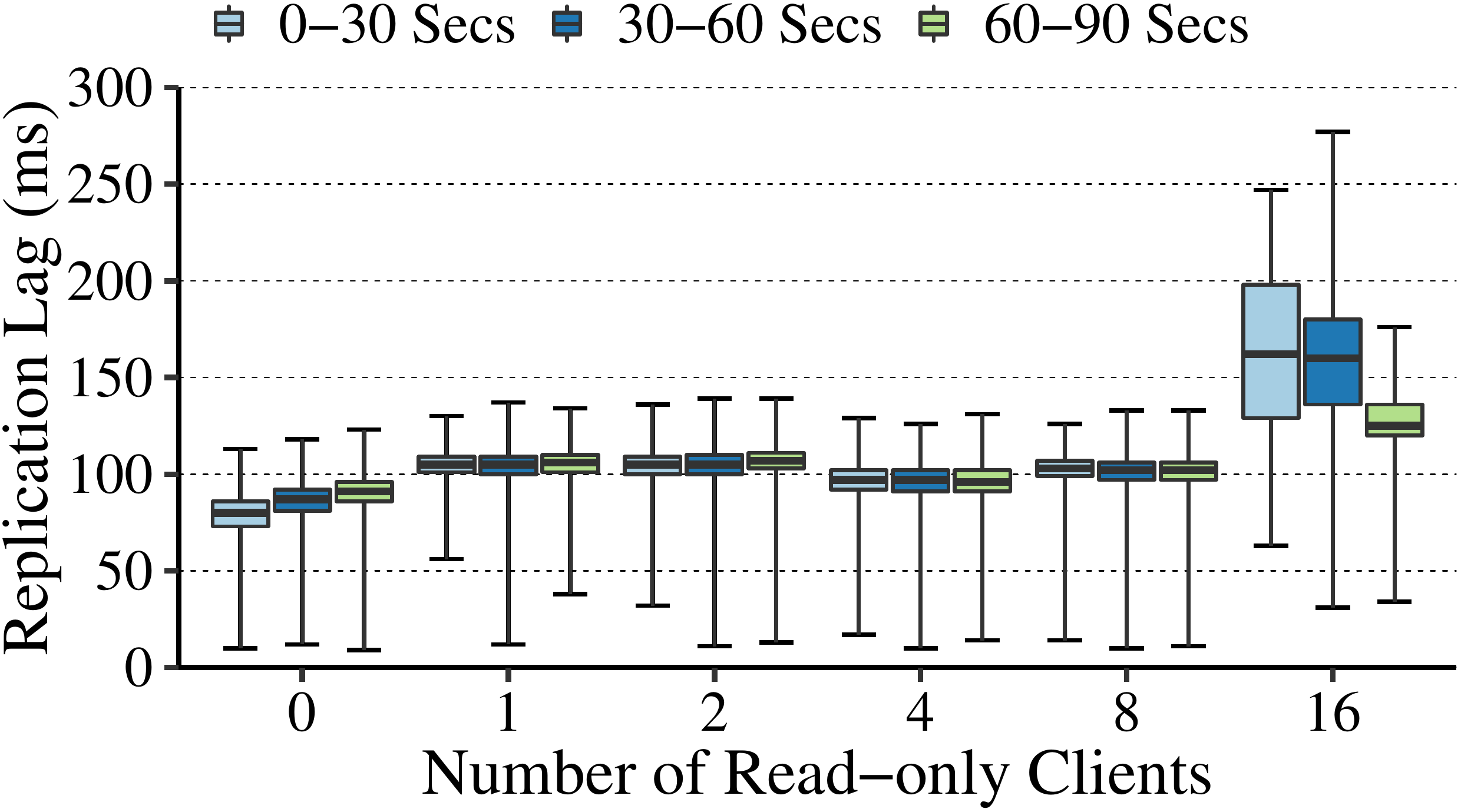}
  \caption{Replication lag of read-write transactions as the number of read-only clients increases. Measurements are split into three 30-second periods.
    Whiskers show the min and max; boxes show the quartiles.}
  \label{fig:read-only-replication-lag}
\end{figure}

Figure~\ref{fig:read-only-replication-lag} plots the distribution of replication
lag measurements for read-write transactions over three consecutive, 30-second
periods with the insert-only workload. The whiskers show the minimum and
maximum, and the boxes show the quartiles. We show one trial; results from other
trials were similar. For each read-write transaction, we measure replication
lag as the difference between when it commits on the primary and when it is
included in the current snapshot. Snapshots were taken every
\SI{10}{\milli\second}. Each read-only transaction executes a random point query
on the table's primary key; queries could select a nonexistent key. We vary
the number of clients from 0 to 16.

The replication lag distributions with zero read-only clients are the baselines.
Replication lag remains bounded despite \sysmyrocks{} blocking \worker{}s
to take snapshots. The remaining plots demonstrate that \sysmyrocks{} provides
bounded replication lag as the load of read-only clients
increases. Median replication lag increases from \SI{87}{\milli\second} with 0
clients to \SI{160}{\milli\second} with 16, and the maximum was always less than
\SI{300}{\milli\second}.

\begin{figure}[t]
  \centering
  \includegraphics[width=0.85\linewidth]{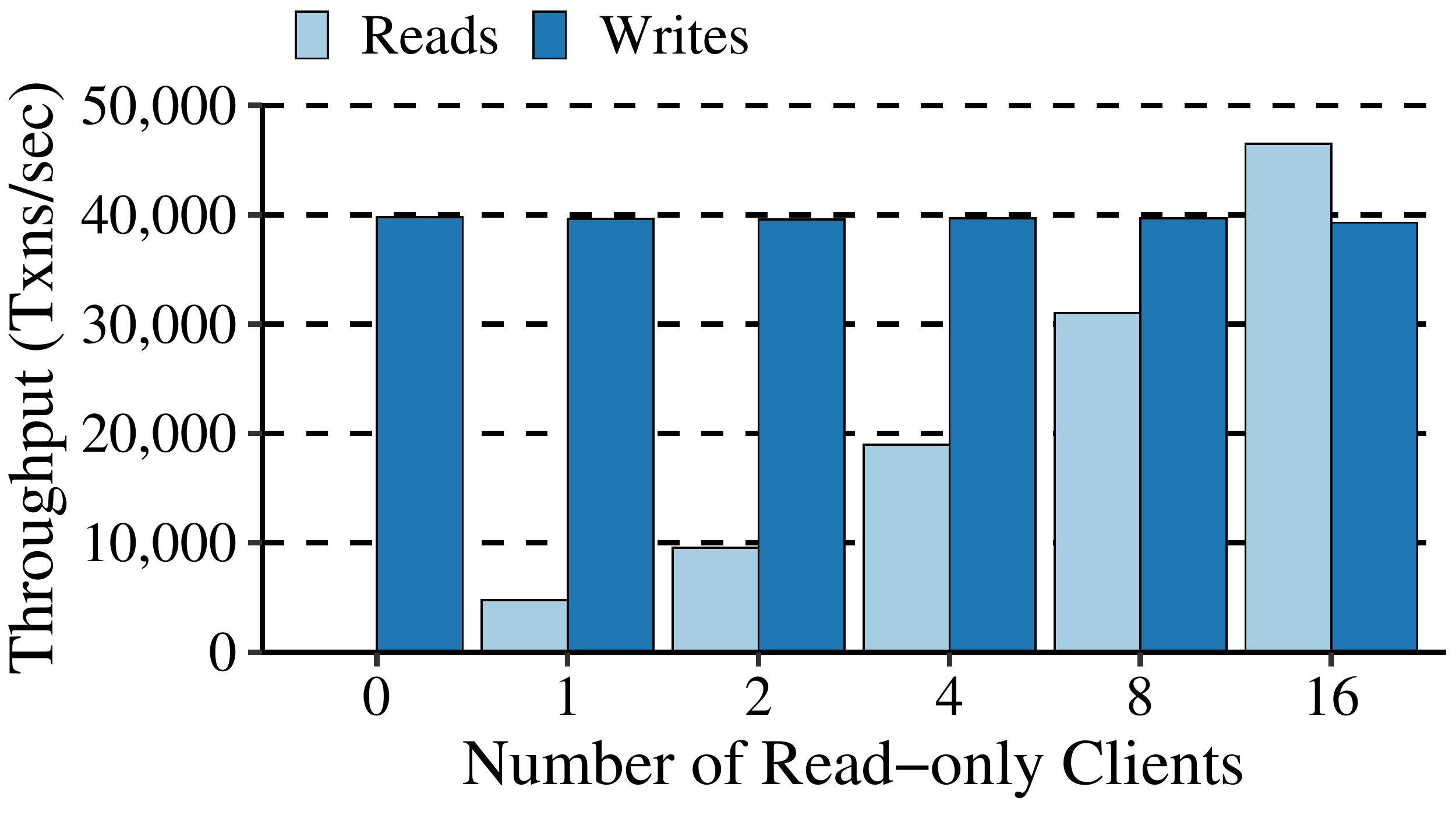}
  \caption{Backup's read-only and read-write transaction throughput as the read-only load increases. \sysmyrocks{} isolates \worker{}s from read-only transactions.}
  \label{fig:read-only-tput}
\end{figure}

Figure~\ref{fig:read-only-tput} plots the backup's throughput of read-only and
read-write transactions for the same experiment.
\sysmyrocks{}'s throughput always matches the primary's. (Note that differences
in the primary's throughput from prior experiments is due to variations in
\myrocks{}'s performance across trials.) Further, \sysmyrocks{} isolates \worker{}s from read-only transactions. With steady write throughput, read-only throughput
increases from \SI{4755}{transactions\per\second} with 1 client to
\SI{46500}{transactions\per\second} with 16, more than doubling the system's total
throughput.


%% file: cicada_impl.tex
\section{\syscicada{}}
\label{sec:cicada}

Unlike \sysmyrocks{}, our implementation of \sys{} in \cicada{}~\cite{lim2017cicada}, an in-memory
multi-version database, is faithful to our design. In this section, we first provide background on \cicada{}. We then describe the implementation of \syscicada{}'s \scheduler{}, \worker{}s, and \snapshotter{}. We conclude with our evaluation, which demonstrates \sys{} can keep up with a modern concurrency control protocol.

\subsection{\cicada{} Background}
\label{sec:cicada:cicada}

\cicada{}'s multi-version storage engine supports reads,
inserts, updates, and deletes. The storage engine is implemented as an array indexed by an internal row ID. (Externally meaningful keys are mapped to row IDs through indices.) Array entries are linked lists of row versions in descending
timestamp order. In addition to a pointer to the next oldest version,
a row version contains data, a status, and
read and write timestamps. The latter three are used by \cicada{}'s concurrency control protocol.

\cicada{} uses a variant of multi-version timestamp ordering~\cite{bernstein-book,lim2017cicada}. Each client thread maintains a local clock. (\cicada{} does not support networked clients.) The local clocks are loosely synchronized and individually return increasing values.

A client uses its clock to assign a unique timestamp to each transaction.
As the transaction executes, each write
creates a new row version, and the transaction's timestamp becomes the version's write timestamp. Further, reading a version updates its
read timestamp to the max of the transaction's timestamp and
the version's current read timestamp. \cicada{} uses the timestamps to check
if a transaction can commit under
serializability~\cite{papadimitriou1979serializability}. Ordering transactions by their timestamps yields a valid serial schedule.


\noindentparagraph{Logging \& Replication.}
Because \cicada{} does not support networking, logging, persistence, or replication~\cite{lim2017cicada}, we emulate primary-backup replication on one server.

To this end, we implement a minimal prototype logger to allow replay of the primary's writes on the backup. The primary only writes logs to memory. After execution and validation but before committing, each client thread logs its changes to a per-thread log. The per-thread logs are coalesced into a single, totally ordered log before the backup's \scheduler{}, \worker{}s, and \snapshotter{} start.

The log is divided into fixed-size segments, each backed by a \SI{2}{\mebi\byte} huge page~\cite{hugetlbfs}. Each segment's header indicates the number of log records it contains. For simplicity, the logger ensures transactions never span segment boundaries.

A client creates a log record for each write in a transaction. Each record contains the following: a table ID, a row ID, the write's timestamp, and a full copy of the row version. Further, the log contains two pieces of metadata to be used by the \scheduler{}, as described below: First, each segment header contains a Boolean \texttt{preprocessed} flag. Second, each record contains an unused, 64-bit \texttt{prev\_timestamp} field.

\subsection{Implementation}
\label{sec:cicada:impl}

\noindentparagraph{Scheduling \& Execution.} For simplicity, we describe the implementation below as if replicating one table. The full implementation supports multiple table by instead keeping a queue for each pair of table and row IDs.

To guarantee monotonic prefix consistency, \syscicada{}'s \scheduler{} must ensure its \worker{}s execute each row's writes in the order they appear in the log. As described in Section~\ref{sec:design:scheduler}, it must construct a FIFO queue of writes to each row. But dynamically allocating and managing these queues prevented the single-threaded \scheduler{} from keeping up with \cicada{}.

The \scheduler{} instead embeds the per-row FIFOs in the log by setting each log record's \texttt{prev\_timestamp} to the timestamp of the write to the same row that immediately precedes it. More specifically, to calculate a record's \texttt{prev\_timestamp}, it maintains a map of row IDs to the write timestamp of the last write to that row (or zero if it is a new row). Thus to process a record, the \scheduler{} reads the value from the map using the log record's row ID, updates the record's \texttt{prev\_timestamp}, and finally updates the map's value with the record's write timestamp. After the \scheduler{} processes all of a segment's records, it sets its \texttt{preprocessed} flag to true.

\Worker{}s are assigned to segments in a round robin order. Once the \scheduler{} finishes processing a segment, the \worker{} begins executing its writes, one for each log record, in three steps: First, it uses the record's \texttt{prev\_timestamp} to see if its write is safe to execute. If \texttt{prev\_timestamp} is equal to the write timestamp of the row version at the head of the storage engine's version list, then the logs record's write should be executed next. Otherwise it is deferred. Second, if the write is safe, the \worker{} allocates a new row version and copies in the necessary data from the log. Third, the new version is installed as the head of row's version list.

Each \worker{} maintains a local FIFO queue of deferred writes and periodically (at the end of each segment) re-checks these writes to see if they are now safe to execute. Thus, \syscicada{}'s \worker{}s maintain a distributed, approximate version of the \scheduler{} queue described in Section~\ref{sec:design:scheduler}.


\noindentparagraph{\Snapshotter{}.} \cicada{}'s storage engine can efficiently implement the API in Table~\ref{tbl:logical-storage-api} because \worker{}s can explicitly assign timestamps to their executed writes. This allows \worker{}s to write to specific snapshots. Further, read-only transactions can execute against the current snapshot simply by using the sequence number $c$ as a timestamp. Their reads will then reflect any previously executed writes with lesser timestamps. The storage engine thus logically contains the current, next, and future snapshots.

This simplifies \syscicada{}'s \snapshotter{}.
To merge the current and next snapshots, it simply advances $c$ to $n$, and replacing the next
snapshot with the future snapshot and creating a new future snapshot occur implicitly when $c$ and $n$ advance.

Before advancing $c$ to $n$, however, the \snapshotter{} must guarantee all writes with timestamps less than or equal to $n$ finish executing. To do this, it cooperates with the \worker{}s.

Each \worker{} maintains a variable $c^\prime$ as one less than the timestamp of the write it most recently executed. Because log records are ordered by timestamp and each \worker{} processes segments in log order, the \worker{} can guarantee it will never execute another write with a timestamp less than or equal to $c^\prime$ (assuming no writes are deferred). Thus, each $c^\prime$ is a upper bound on $c$.    
When a \worker{} defers a write, it also defers updating $c^\prime$ until the write executes.

The \snapshotter{}'s implementation is thus simple; in a separate thread, it periodically calculates a new $n$ as the minimum across all $c^\prime$ and then advances $c$ to $n$. Since each $c^\prime$ has only one reader and one writer, no coordination or atomic instructions are needed. With x86's total store order, the \snapshotter{} may read a stale copy of a \worker{}'s $c^\prime$, but this does not violate correctness.

\subsection{Evaluation}
\label{sec:cicada:evaluation}

Our evaluation of \syscicada{} explores whether our design can keep up with advanced concurrency control protocols, both on realistic and synthetic workloads.

\noindentparagraph{Experimental Setup.} Experiments ran on the same machines as described in Section~\ref{sec:myrocks-evaluation}. The primary's threads and the backup's \scheduler{}, \worker{}s, and \snapshotter{} are all pinned different cores.

Because logging slows \cicada{}, we compare against \cicada{}'s performance without logging, which is an upper bound on that with logging enabled. We again use the optimal number of primary threads and backup \worker{}s, with the latter never exceeding the former. We ran experiments 5 times and report the median. Error bars show the minimum and maximum.

The workloads are the same as in Section~\ref{sec:myrocks-evaluation}, and we re-implement \kuafu{} in \cicada{} for fair comparison. In fact, we implement two versions of \kuafu{}, one optimized for low contention and the other for high contention. (The latter matches the published pseudocode~\cite{hong2013kuafu}.) We report whichever achieves better performance.

\noindentparagraph{\syscicada{} Prevents Disasters.} \cicada{}'s concurrency control is much better than \myrocks{}'s at handling contention. Thus, we expect significantly more potential for replication lag (and reliability issues).

Our results support this; \syscicada{} is necessary to prevent replication lag with the standard \SI{50}{\percent}-\SI{50}{\percent} NewOrder-Payment workload after applying the same optimizations described in Section~\ref{sec:myrocks-evaluation}. \cicada{} achieves \SI{716950}{transactions\per\second} while \kuafu{} manages only \SI{596310}{transactions\per\second}, lagging by about \SI{17}{\percent}. \syscicada{} easily keeps up, committing \SI{1062533}{transactions\per\second}. The results on the unoptimized workload are similar. The primary's throughput is lower, but \kuafu{} still lags by about \SI{13000}{transactions\per\second} (\SI{3}{\percent}).

\begin{figure}[t]
  \includegraphics[width=0.85\linewidth]{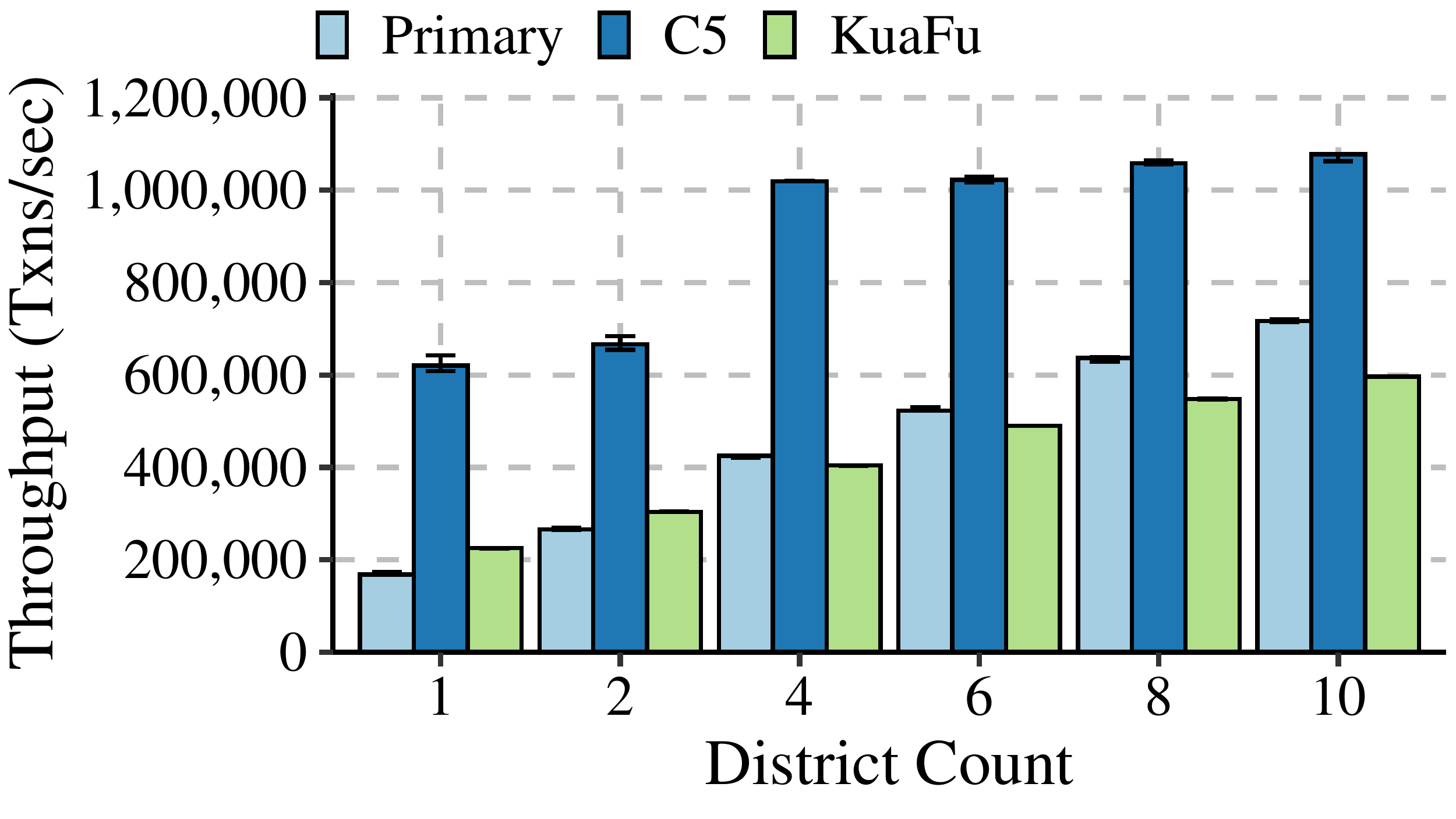}
  \caption{Primary and backup's throughputs on NewOrder-Payment workload with varying number of districts.}
  \label{fig:cicada-tpcc}
\end{figure}

We also explore how the primary and two backups behave under varying levels of contention. Figure~\ref{fig:cicada-tpcc} compares the throughputs of \cicada{}, \syscicada{}, and \kuafu{} on a \SI{50}{\percent}-\SI{50}{\percent} workload as we vary the number of districts from 10 (the standard setting) to 1. As contention increases (i.e., the number of districts decreases), \kuafu{} continues to lag until 4 districts. But below that, the additional contention harms \cicada{}'s throughput more than \kuafu{}'s by causing significantly higher abort rates (up to about \SI{75}{\percent}). As a result, with fewer districts, \kuafu{} keeps up.

To help confirm \kuafu{} lags due to the constraints imposed on its execution, we re-ran the experiment above but disabled its \scheduler{}'s calculation of transaction-granularity constraints. For each number of districts, we compared \kuafu{}'s throughput while using the same number of \worker{}s as shown in Figure~\ref{fig:cicada-tpcc}, and in all cases, \kuafu{} no longer lagged. For example, with 10 districts and 6 \worker{}s, \kuafu{}'s throughput nearly doubles from \SI{596310}{} to \SI{1101491}{transactions\per\second} when its execution is unconstrained. This far exceeds the primary's throughput of \SI{716950}{transactions\per\second}.

These results highlight the complexity of predicting when existing \ccc{} protocols will be sufficient to avoid replication lag. As shown in Figure~\ref{fig:cicada-tpcc}, \syscicada{} always keeps up and thus removes the potential for disasters.

\noindentparagraph{\syscicada{} Always Keeps Up.} We again validate \syscicada{} using the insert-only and adversarial workloads.

On the insert-only workload, \cicada{} achieves its best performance with transactions of 16 inserts each, amortizing its per-transaction overhead. On this workload, \cicada{} inserts about 87M rows/s with 20 threads. \kuafu{} and \syscicada{} both keep up. The former commits 96M rows/s with 12 \worker{}s and the latter 99M rows/s with 10 \worker{}s. In both cases, even on a workload designed to expose bottlenecks, the \scheduler{}s provide sufficient performance.


\begin{figure}[t]
  \centering
  \includegraphics[width=0.85\linewidth]{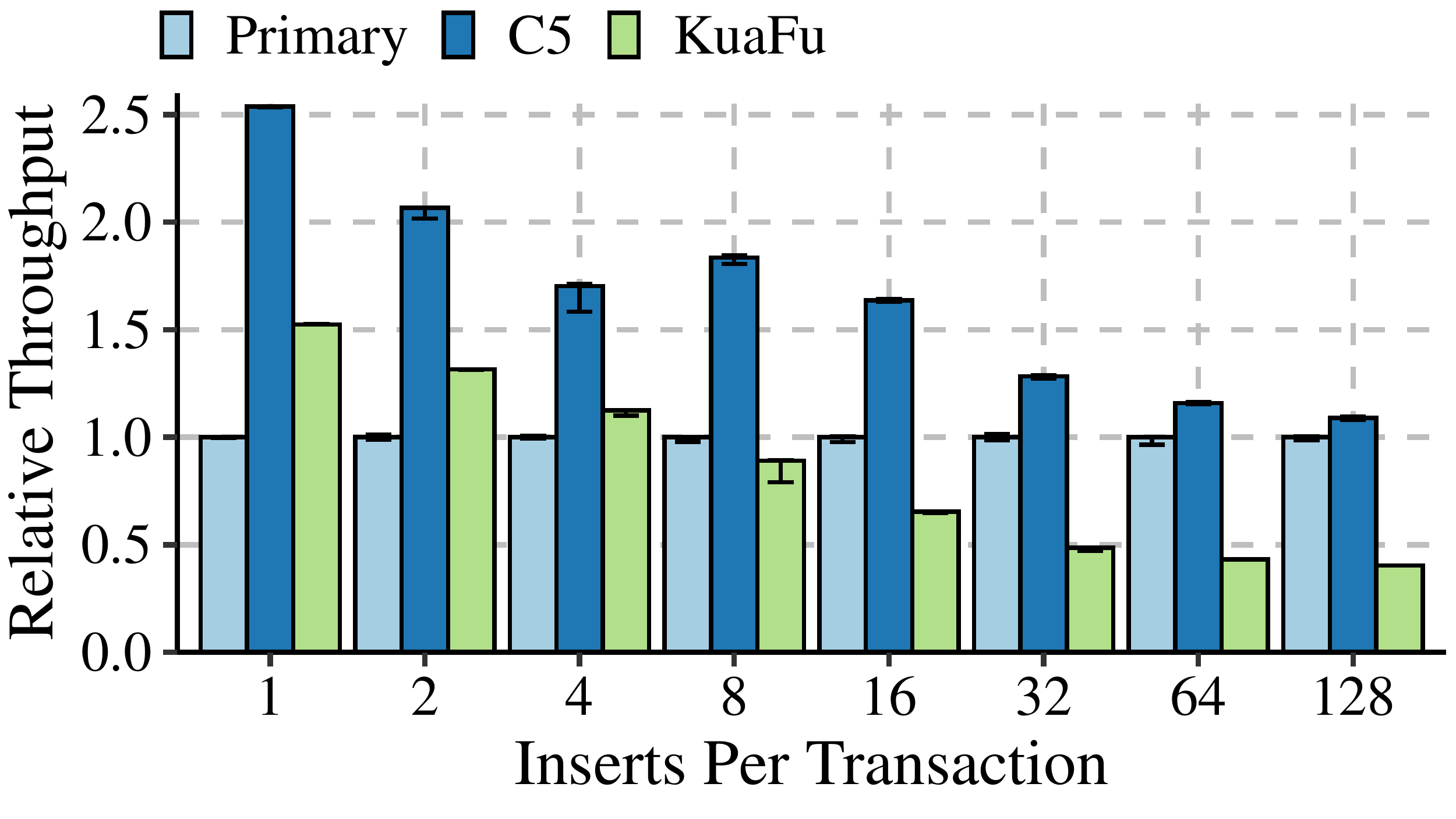}
  \caption{Backup's throughput relative to the primary's for the adversarial workload.}
  \label{fig:cicada-adversarial}
\end{figure}

Figure~\ref{fig:cicada-adversarial} compares each \ccc{} protocol's performance to \cicada{}'s
on the adversarial workload. We plot the backup's throughput relative to the
primary's median as we vary the number of non-conflicting inserts per transaction.

\syscicada{} mirrors the primary and executes the
non-conflicting inserts in parallel, so it always keeps up. This advantage is especially evident as the number of inserts per transaction increases from 4 to 8. With more parallel work per transaction, \syscicada{} leverages additional \worker{}s and its relative throughput actually increases.  

On the other hand, the primary's advantage over \kuafu{}
increases with the number of inserts
per transaction. With 128 inserts per transaction, \kuafu{}'s throughput is just \SI{40}{\percent} of the primary's.


%% file: deployment.tex
\section{Deployment Experience}
\label{sec:deployment}

\myrocks{} databases are deployed and used inside the globally
distributed datacenters at \fbprom{}. Each shard of their social graph uses
asynchronous primary-backup replication. Further, their internal cloud provides
asynchronously replicated, multi-tenant \myrocks{} instances.
Before deploying \sysmyrocks{}, ensuring
short replication lag throughout their infrastructure was a persistent
challenge.

A version of \sysmyrocks{} has been deployed in production since mid-2017, and full
deployment finished in early 2019. Deploying \sysmyrocks{} at \thefb{} has notably
improved the reliability of their applications built atop asynchronous
primary-backup databases. Since its deployment, anecdotally, the number of
complaints about replication lag by on-call engineers
drastically decreased.

Figure~\ref{fig:fb-lag-plot} shows one example of significant lag fixed by
\sysmyrocks{}. It plots throughput over time for one shard. During daily periods of
high insert load, the primary's throughput exceeded the backup's, and lag grew to over two hours both with MySQL 5.6's default, single-threaded
\ccc{}~\cite{mysql,myRocks} and \fb{}'s earlier table-granularity protocol. After
the load spike ends, both protocols took two hours to return lag to zero.
\sysmyrocks{} eradicated the issue, keeping lag below three seconds.

Similarly, after live videos were deployed, popular videos became
significant sources of contention. Like in the motivating example in
Section~\ref{sec:bg:example}, all comments on a video caused writes to a
single row. With prior \ccc{} protocols, popular videos caused
significant replication lag. Although this problem was initially fixed by
batching comment-update requests, \sysmyrocks{} would have avoided the problem
entirely.

\begin{figure}[t]
  \centering
  \includegraphics[width=0.85\linewidth]{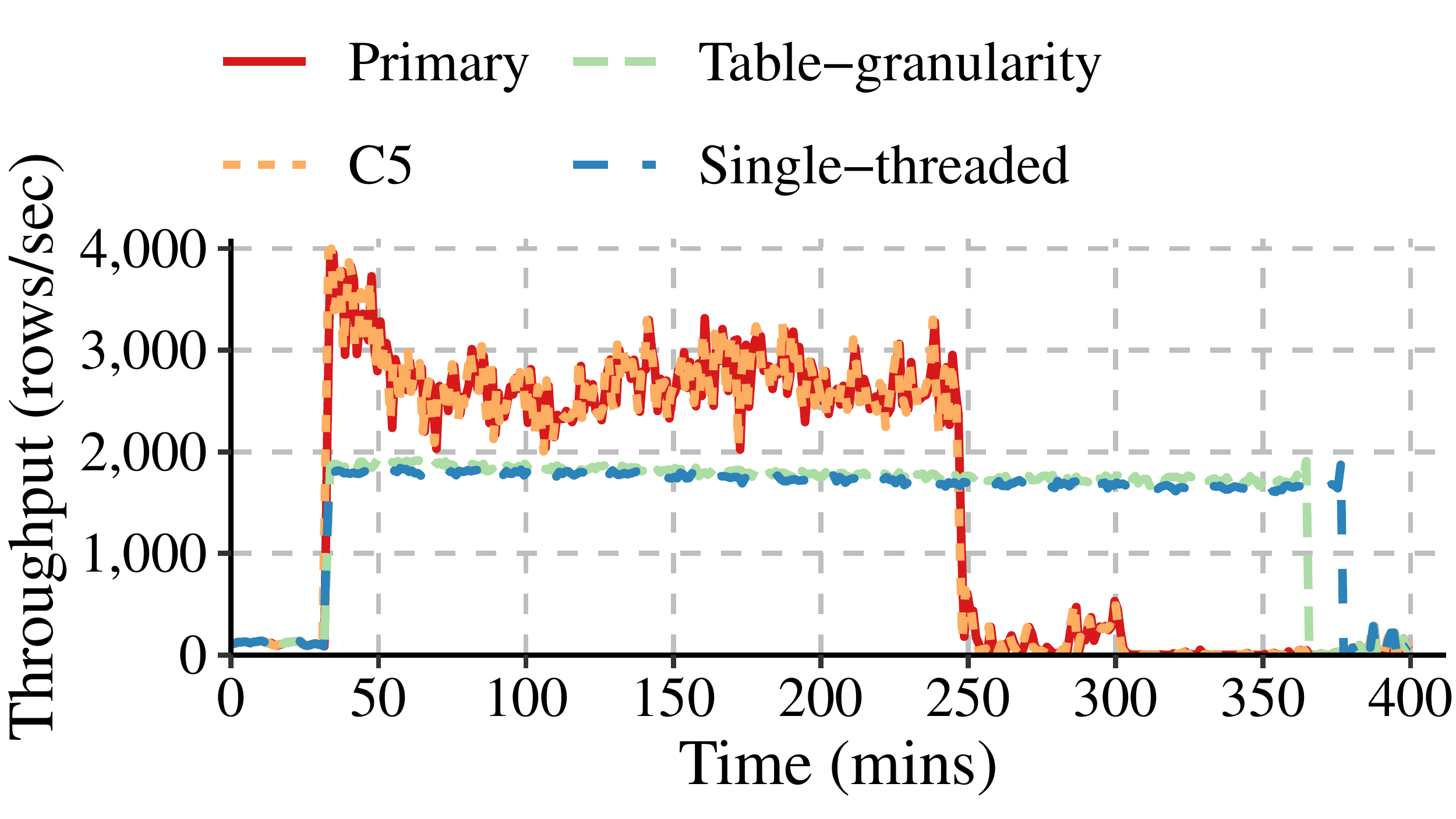}
  \caption{Replication lag at \fb{} used to reach
    nearly 2 hours every day and took 2 hours to recover. With \sysmyrocks{}, lag
    remains below 3 seconds.}
  \label{fig:fb-lag-plot}
\end{figure}

Solving replication lag also had secondary benefits: First, deploying \sysmyrocks{}
revealed bottlenecks in downstream systems,
which have since been fixed by \thefb{}'s engineers. Second, \thefb{} uses
cross-region replication. Multiple copies of the data exist on servers within a
primary region, and writes asynchronously replicate to backups in other
regions. Short lag reduces the number of times that data must be fetched from
other regions to satisfy read-your-writes consistency for clients with recent
writes. Further, if the entire primary region fails while all machines in the
backup region are lagging, some unreplicated writes may be lost. With
\sysmyrocks{}, the probability of user data loss and the magnitude of such loss if
it occurs are both significantly reduced. Finally, \sysmyrocks{} eliminated a
noisy-neighbor problem experienced by applications deployed on \thefb{}'s
internal cloud. Applications using a multi-tenant \myrocks{} instance
previously would sometimes experience replication lag caused by others sharing the instance.


%% file: related.tex
\section{Related Work}




\noindentparagraph{Deterministic Concurrency Control.} Deterministic concurrency control protocols~\cite{faleiro2015multiversion,whitney1997without, faleiro2014lazy,faleiro2017visibility,thomson2012calvin,qin2021caracal} ensure that database
state is a deterministic function of the input log. \sys{}'s processing of
writes from the primary is inspired by such protocols. These
include the up-front resolution of write-write conflicts prior to executing
them~\cite{faleiro2015multiversion} and its representation of
permissible execution schedules of writes~\cite{whitney1997without, faleiro2014lazy}. Databases employing deterministic
concurrency control, however, do not designate a single replica as a primary and others as
backups. They instead employ active replication~\cite{thomson2010case,
  thomson2012calvin}. \sys{}, on the other hand, is applicable to
primary-backup systems where the primary is non-deterministic.

\noindentparagraph{Database Recovery \& Replication.} Database
recovery~\cite{mohan1992aries,lee2001differential,cha2004ptime,schwalb2014hyrise,zheng2014recovery}
and
replication~\cite{mohan1993remote,minhas2011remusdb,oracle2014eagerRep,qin2017replication,zamanian2019activeMem}
are two problems closely related to \ccc{}. In database recovery, changes to the
database are logged and stored on stable storage. If a database fails, a recovery
protocol creates a new copy of the database. Database replication extends
database recovery to reduce recovery time by replicating and applying changes to
the backup while the primary executes transactions. If the primary fails, the
backup executes a synchronization protocol to bring it into a consistent state
before processing new transactions. But database replication (and thus recovery)
is simpler than \ccc{} because backups do not serve read-only transactions, so
the backup only needs to be prefix-consistent before processing new
transactions. On the other hand a \ccc{} must always be able to serve read-only transactions from a prefix-consistent state.

\noindentparagraph{\CCC{}.} No existing asynchronous or semi-synchronous \ccc{} protocol
can guarantee bounded replication lag. Synchronous \ccc{}
protocols~\cite{elnikety2004generalized} trivially guarantee it because
the primary and backup coordinate before a transaction commits. But synchronous
protocols reduce the primary's performance. Thus, asynchronous and
semi-synchronous are more widely
deployed~\cite{existential2015sosp,instagram2015scaling,espresso2013sigmod,verbitski2018aurora,antonopoulos2019socrates}.

Transaction-~\cite{hong2013kuafu,king1991remote,oracle2001txnscheduler,mysql8writeset}
and
page-granularity~\cite{verbitski2018aurora,oracle2019adg,antonopoulos2019socrates}
protocols cannot guarantee bounded replication lag. By similar reasoning,
coarser granularity protocols, such as those using groups of
transactions~\cite{mysqlgroupcommit,mariadb10,oracle2019ggate}, cannot either.

To the best of our knowledge, Query Fresh~\cite{wang2017query} is the only existing row-granularity \ccc{} protocol. But to reduce its \worker{}s' processing while guaranteeing monotonic
prefix consistency, read-only transaction threads instantiate the backup's copy
of the database from the copied log. This lazy instantiation is serialized for
the entire read-only transaction, which may add significant latency. Further,
read-only transaction threads optimistically update the database and will abort
if multiple threads try to update the same row concurrently, further increasing
latency for higher contention workloads. Finally, Query Fresh's lazy instantiation of the backup's database can cause arbitrarily large replication lag even using single-key transactions, but we omit details here due to space constraints.

%% file: conclusion.tex
\section{Conclusion}
\label{sec:conclusion}

We presented \sys{}, the first cloned concurrency protocol to
provide bounded replication lag. \sys{} comprises three
parts: a \scheduler{}, \worker{}s, and a \snapshotter{}. \sys{} is backed by
multiple theoretical results showing the necessity of its
row-granularity protocol. We also presented two implementations: \sysmyrocks{} and \syscicada{}. The former is backward-compatible with MyRocks and deployed at \fb{}, while the latter faithfully implements our design. We demonstrated experimentally they always keep up with their primaries.
